\RequirePackage{rotating}
\documentclass[acmsmall]{acmart}

\usepackage{amsthm}

\usepackage{breqn}
\usepackage{algorithm}
\usepackage{algpseudocode}

\usepackage{epsfig}
\usepackage{epstopdf}
\usepackage{subfigure}
\usepackage{multirow}
\usepackage{comment}
\usepackage{listings}
\usepackage{xcolor}
\usepackage{soul}
\usepackage{xspace}
\usepackage{url}
\usepackage{rotating}
\usepackage{enumitem}
\usepackage{tabularx}

\newcommand{\bigO}[1]{$\mathcal{O}\paren{#1}$\xspace}
\newcommand{\bigOmega}[1]{\mbox{\mbox{$\Omega$}$\paren{#1}$}\xspace}
\newcommand{\bigTheta}[1]{\mbox{\mbox{$\Theta$}$\paren{#1}$}\xspace}
\newcommand{\paren}[1]{\left(  #1 \right)}
\newcommand{\edge}[1]{(  #1 )}

\newtheorem{theorem}{Theorem}

\newtheorem{lemma}{Lemma}

\newtheorem{definition}[theorem]{Definition}

\usepackage{amsmath,amssymb,amsfonts}
\usepackage{graphicx}
\usepackage{textcomp}
\AtBeginDocument{%
  \providecommand\BibTeX{{%
    \normalfont B\kern-0.5em{\scshape i\kern-0.25em b}\kern-0.8em\TeX}}}

\setcopyright{none}
\copyrightyear{2024}
\acmYear{2024}
\acmDOI{XXXXXXX.XXXXXXX}

\begin{document}

\title{Cover Edge-Based Novel Triangle Counting}

\author{David A. Bader}
\email{bader@njit.edu}
\author{Fuhuan Li}
\email{fl28@njit.edu}
\author{Zhihui Du}
\email{zhihui.du@njit.edu}
\author{Palina Pauliuchenka}
\email{pp272@njit.edu}
\author{Oliver Alvarado Rodriguez}
\email{oaa9@njit.edu}
\affiliation{%
  \institution{New Jersey Institute of Technology}
  \city{Newark}
  \state{New Jersey}
  \country{USA}
  \postcode{07102}
}

\author{Anant Gupta}
\affiliation{%
  \institution{John P. Stevens High School}
  \city{Edison}
  \country{USA}}

\author{Sai Sri Vastav Minnal}
\affiliation{%
  \institution{Edison Academy Magnet School}
  \city{Edison}
  \country{USA}
}

\author{Valmik Nahata}
\affiliation{%
 \institution{New Providence High School}
 \city{New Providence}
 \state{New Jersey}
 \country{USA}}

\author{Anya Ganeshan}
\affiliation{%
  \institution{Bergen County Academies}
  \city{Hackensack}
  \state{New Jersey}
  \country{USA}}

\author{Ahmet Gundogdu}
\affiliation{%
  \institution{Paramus High School}
  \city{Paramus}
  \state{New Jersey}
  \country{USA}}

\author{Jason Lew}
\affiliation{%
  \institution{New Jersey Institute of Technology}
  \city{Newark}
  \state{New Jersey}
  \country{USA}
  \postcode{07102}
}
\email{jl247@njit.edu}

\renewcommand{\shortauthors}{Bader~\emph{et al.}}

\begin{abstract}
Listing and counting triangles in graphs is a key algorithmic kernel for network analyses, including community detection, clustering coefficients, k-trusses, and triangle centrality. 
In this paper, we propose the novel concept of a cover-edge set that can be used to find triangles more efficiently. Leveraging the breadth-first search (BFS) method, we can quickly generate a compact cover-edge set.  Novel sequential and parallel triangle counting algorithms that employ cover-edge sets are presented. The novel sequential algorithm performs competitively with the fastest previous approaches on both real and synthetic graphs, such as those from the Graph500 Benchmark and the MIT/Amazon/IEEE Graph Challenge. We implement 22 sequential algorithms for performance evaluation and comparison. At the same time, we employ OpenMP to parallelize 11 sequential algorithms, presenting an in-depth analysis of their parallel performance. Furthermore, we develop a distributed parallel algorithm that can asymptotically reduce communication on massive graphs.  In our estimate from massive-scale Graph500 graphs, our distributed parallel algorithm can reduce the communication on a scale~36 graph by 1156x and on a scale~42 graph by 2368x. Comprehensive experiments are conducted on the recently launched Intel Xeon 8480+ processor and shed light on how graph attributes, such as topology, diameter, and degree distribution, can affect the performance of these algorithms.
\end{abstract}



\maketitle

\section{Introduction}
Triangle listing and counting is a highly-studied problem in computer science and is a key building block in various graph analysis techniques such as clustering coefficients \cite{watts1998collective}, k-truss \cite{cohen2008trusses}, and triangle centrality \cite{burkhardt2021triangle}, \cite{li2021graphblas}.  The significance of triangle counting is evident in its application in high-performance computing benchmarks like Graph500 \cite{graph500} and the MIT/Amazon/IEEE Graph Challenge \cite{GraphChallenge}, as well as in the design of future architecture systems (e.g., IARPA AGILE \cite{slides_on_AGILE}). 

There are at most $\binom{n}{3} = \bigTheta{n^3}$ triangles in a graph $G = (V, E)$ with $n=|V|$ vertices and $m=|E|$ edges. The na\"ive approach using triply-nest loops to check if each triple $(u, v, w)$ forms a triangle takes \bigO{n^3} time and is inefficient for sparse graphs.  It is well-known that listing all triangles in G is $\bigOmega{m^{\frac{3}{2}}}$ time \cite{itai1978finding, latapy2007practical}. To enhance the performance of triangle counting, Cohen \cite{cohen2009graph} introduced a novel map-reduce parallelization technique that generates \emph{open wedges} between triples of vertices in the graph. It determines whether a closing edge exists to complete a triangle, thus avoiding the redundant counting of the same triangle while maintaining load balancing. Many parallel approaches for triangle counting \cite{pearce2017triangle,ghosh2020tric} partition the sparse graph data structure across multiple compute nodes and adopt the strategy of generating open wedges, which are sent to other compute nodes to determine the presence of a closing edge. Consequently, the communication time for these open wedges often dominates the running time of parallel triangle counting.

In this paper, we propose a novel approach that efficiently identifies all triangles using a reduced set of edges known as a cover-edge set. By leveraging the cover-edge-based triangle counting method, unnecessary edge checks can be skipped while ensuring that no triangles are missed. This significantly reduces the number of computational operations compared to existing methods. 

The main contributions of this paper are

\begin{itemize}
\item A novel triangle counting algorithm, \emph{Cover-Edge Triangle Counting (CETC)}, is proposed based on a new concept \emph{Cover-Edge Set}. The essential idea is that we can identify all triangles from a significantly reduced cover-edge set instead of the complete edge set. A simple breadth-first search (BFS) is used to orient the graph's vertices into levels and to generate the cover-edge set. 
\item Various sequential variants of the \emph{CETC} that combine the techniques of cover-edge, forward algorithm, and hashing are developed. Furthermore, the parallel implementations of \emph{CETC} on both shared-memory \emph{(CETC-SM)}, and distributed-memory \emph{(CETC-DM)} are introduced. 
\item Freely-available, open-source software for more than 22 sequential triangle counting algorithms and 11 OpenMP parallel algorithms in the C programming language.
\item A comprehensive experimental study of implementations of the proposed novel triangle counting algorithms on real and synthetic graphs with the comparison against other existing algorithms.
\end{itemize}

\section{Notations and Definitions}
\label{subsec:notation}
Let $G = (V, E)$ be an undirected graph with $n=|V|$ vertices and $m=|E|$ edges. A \emph{triangle} in the graph is a set of three vertices $\{v_a, v_b, v_c\} \subseteq V$ such that $\{\edge{v_a, v_b}, \edge{v_a, v_c}, \edge{v_b, v_c}\} \subseteq E$. We will use $N(v) = \{ u|u \in V \wedge (\edge{u,v} \in E)\}$ to denote the \emph{neighbor set} of vertex $v \in V$.
The degree of vertex $v \in V$ is $d(v) = |N(v)|$, and $d_{\text{max}}$ is the maximal degree of a vertex in graph $G$.

With these notations, the total number of triangles in graph $G$ is denoted as $|\Delta(G)|$. Specifically, $\Delta(G)=\{(u,v,w)| u,v,w$ are different vertices of $V$ and $\edge{u,v}, \edge{v,w}, \edge{w,u}$ are edges of $E\}$.

The triangle counting problem can be expressed in two ways based on edges and vertices:

\begin{itemize}
    \item For any edge $\edge{u,v} \in E$, the number of triangles including $\edge{u,v}$ is $|\Delta \edge{u,v}|$, where $\Delta \edge{u,v} = N(u) \cap N(v)$. Since each triangle edge will count the same triangle and we will count both $\Delta \edge{u,v}$ and $\Delta \edge{v,u}$, the total triangles are computed as $|\Delta(G)|=\frac{\sum_{\edge{u,v}\in E }|\Delta \edge{u,v}|}{6}$, using the edge-iteration-based method.
    \item For any vertex $v \in V$, the number of triangles including $v$ is $|\Delta(v)|$, where $\Delta(v) = \{ \edge{u,w} \,|\, u, w \in N(v) \land \edge{u,w} \in E\}$. The total triangles are computed as $|\Delta(G)|=\frac{\sum_{v\in V }|\Delta(v)|}{6}$, using the vertex-iteration-based method.
\end{itemize}

\section{Related Work}
\label{sec:related}
\subsection{Existing Sequential Algorithms}
For triangle counting, the obvious algorithm is brute-force search (see Alg.~\ref{alg:triplets}), enumerating over all \bigTheta{n^3} triples of distinct vertices, and checking how many of these triples are triangles. There are faster algorithms that require an adjacency matrix for the input graph representation and use fast matrix multiplication, such as the work of Alon, Yuster, and Zwick \cite{alon1997finding}. Indeed, if $A$ is the adjacency matrix of $G$, for any vertex $v$, the value $A_{vv}^3$ on the diagonal of $A^3$ is twice the number of triangles to which $v$ belongs. So the number of trianlges is $\frac{1}{6}\sum tr(A^3)$. Triangle counting problems can therefore be solved in \bigO{n^{1.5}}, where $\omega < 2.732$ is the fast matrix product exponent \cite{alman2021refined} \cite{williams2023new}. Alon~\emph{et al.} \cite{alon1997finding} also show that it is possible to solve triangle counting problem in \bigO{m^{\frac{2 \omega}{\omega + 1}}} $\subset$ \bigO{m^{1.41}} time. However, the implementation is infeasible for large, sparse graphs, and certain matrix multiplication methods fall short of listing all the triangles. For these reasons, despite their evident theoretical strength, these algorithms have limited practical impact.

\begin{algorithm}[htbp]
\footnotesize
\caption{Triples}
\label{alg:triplets}
\begin{algorithmic}[1]
\Require{Graph $G = (V, E)$}
\Ensure{Triangle Count $T$}
\State $T \leftarrow 0$  
\State $\forall u \in V$
\State \hspace{8pt} $\forall v \in V$
\State \hspace{16pt} $\forall w \in V$
\State \hspace{24pt} if $(u,v) \in E \land (v,w) \in E \land (u,w) \in E$
\State \hspace{32pt} $T \leftarrow T + 1$
\State return $T/6$
\end{algorithmic}
\end{algorithm}

Another category of fundamental problem formulation is called subgraph query, which aims to identify instances of a triangle subgraph within the input graph. It's crucial to emphasize that determining the presence of a specific subgraph in a graph is an NP-hard problem. While various methods, including the backtracking strategy \cite{ullmann1976algorithm}, have been introduced, they are not preferred choices for triangle counting problem, particularly for large-scale graphs. 

Latapy \cite{latapy2007practical} provides a survey on triangle counting algorithms for very large, sparse graphs. One of the earliest algorithms, \emph{tree-listing}, published in 1978 by Itai and Rodeh \cite{itai1978finding} first finds a rooted spanning tree of the graph. After iterating through the non-tree edges and using criteria to identify triangles, the tree edges are removed and the algorithm repeats until no edges are remaining (see Alg.~\ref{alg:tree-listing}). This approach takes \bigO{m^{\frac{3}{2}}} time (or \bigO{n} for planar graphs).

\begin{algorithm}[htbp]
\footnotesize
\caption{Tree-listing (IR) \cite{itai1978finding}}
\label{alg:tree-listing}
\begin{algorithmic}[1]
\Require{Graph $G = (V, E)$}
\Ensure{Triangle Count $T$}
\State $T \leftarrow 0$ 
\State while $E$ is not empty
\State \hspace{8pt} $K \leftarrow$ Covering tree($G$) 
\State \hspace{8pt} $\forall (u,v) \in E\ \land (u,v) \notin K$
\State \hspace{16pt} if $(\text{parent}(u),v) \in E$
\State \hspace{24pt} $T \leftarrow T + 1$
\State \hspace{16pt} elif $(\text{parent}(v),u) \in E$
\State \hspace{24pt} $T \leftarrow T + 1$
\State \hspace{8pt} $E \leftarrow E-K$
\State return $T/2$
\end{algorithmic}
\end{algorithm}

The most common triangle counting algorithms in the literature include 
\emph{vertex-iterator} \cite{itai1978finding}, \cite{latapy2007practical}
and \emph{edge-iterator} \cite{itai1978finding}, \cite{latapy2007practical} approaches that run in \bigO{m \cdot d_{max}}.

\begin{algorithm}[htbp]
\footnotesize
\caption{Vertex-Iterator \cite{itai1978finding}, \cite{latapy2007practical}}
\label{alg:node}
\begin{algorithmic}[1]
\Require{Graph $G = (V, E)$}
\Ensure{Triangle Count $T$}
\State $T \leftarrow 0$   
\State $\forall u \in V$  
\State \hspace{8pt} $\forall v \in N(u)$ 
\State \hspace{16pt} $X = Intersection(N(u),N(v))$
\State \hspace{16pt} $T \leftarrow T + X$
\State return $T/6$
\end{algorithmic}
\end{algorithm}

In vertex-iterator (see Alg.~\ref{alg:node}), for each vertex $u \in V$, the algorithm examines the adjacency list $N(v)$ of each vertex $v \in N(u)$. If there is vertex $w$ in the intersection of $N(u)$ and $N(v)$, then the triplet $(u,v,w)$ forms a triangle. Arifuzzaman~\emph{et al.} \cite{arifuzzaman2019} study modifications of the vertex-iterator algorithm based on various methods for vertex ordering.

\begin{algorithm}[htbp]
\footnotesize
\caption{Edge-Iterator \cite{itai1978finding}, \cite{latapy2007practical}}
\label{alg:edge}
\begin{algorithmic}[1]
\Require{Graph $G = (V, E)$}
\Ensure{Triangle Count $T$}
\State $T \leftarrow 0$   
\State $\forall (u,v) \in E$   
\State \hspace{8pt} $X = Intersection(N(u), N(v))$
\State \hspace{8pt} $T \leftarrow T + X$
\State return $T/6$
\end{algorithmic}
\end{algorithm}

In edge-iterator (see Alg.~\ref{alg:edge}), each edge $(u, v)$ in the graph is examined, and the intersection of $N(u)$ and $N(v)$ is computed to find triangles. A common optimization is to use a \emph{direction-oriented} approach that only considers edges $(u, v)$ where $u< v$.  
The variants of edge-iterator are often based on the algorithm used to perform $Intersection(N(u), N(v))$. When the two adjacency lists are sorted, then \emph{MergePath} and \emph{BinarySearch} can be used. MergePath performs a linear scan through both lists counting the common elements. Makkar, Bader, and Green \cite{makkar2017} give an efficient MergePath algorithm for GPU. Mailthody~\emph{et al.} \cite{mailthody2018} use an optimized two-pointer intersection (MergePath) for set intersection.
BinarySearch, as the name implies, uses a binary search to determine if each element of the smaller list is found in the larger list. \emph{Hash} is another method for performing the intersection of two sets and it does not require the adjacency lists to be sorted. A typical implementation of \emph{Hash} initializes a Boolean array of size $m$ to all false. Then, positions in \emph{Hash} corresponding to the vertex values in $N(u)$ are set to true. Then $N(v)$ is scanned, looking up in $\bigTheta{1}$ time whether or not there is a match for each vertex. Chiba and Nishizeki published one of the earliest edge iterators with hashing algorithms for triangle finding in 1985 \cite{chiba1985}. The running time is \bigO{a(G) m}, where $a(G)$ is defined as the arboricity of $G$, which is upper-bounded $a(G) \leq \lceil (2m + n) ^{\frac{1}{2}} / 2 \rceil$ \cite{chiba1985}. In 2018, Davis rediscovered this method, which he calls \texttt{tri\_simple} in his comparison with SuiteSparse GraphBLAS \cite{davis2018HPEC}. Mowlaei \cite{Mowlaei2017} gave a variant of the edge-iterator algorithm that uses vectorized sorted set intersection and reorders the vertices using the reverse Cuthill-McKee heuristic.

In 2005, Schank and Wagner \cite{schank2005finding,schank2007} designed a fast triangle counting algorithm called \emph{forward} (see Alg.~\ref{alg:forward}) that is a refinement of the edge-iterator approach. Instead of intersections of the full adjacency lists, the \emph{forward} algorithm uses a dynamic data structure $A(v)$ to store a subset of the neighborhood $N(v)$ for $v \in V$. Initially, each set $A()$ is empty, and after computing the intersection of the sets $A(u)$ and $A(v)$ for each edge $(u, v)$ (with $u<v$), $u$ is added to $A(v)$. This significantly reduces the size of the intersections needed to find triangles. The running time is \bigO{m \cdot d_{\mbox{max}}}. However, if one reorders the vertices in decreasing order of their degrees as a \bigTheta{n \log n} time pre-processing step, the forward algorithm's running time reduces to \bigO{m^{\frac{3}{2}}}. Ortmann and Brandes \cite{ortmann2014} survey triangle counting algorithms, create a unifying framework for parsimonious implementations, and conclude that nearly every triangle listing variant is in \bigO{m \cdot a(G)}. 

\begin{algorithm}[htbp]
\footnotesize
\caption{Forward Triangle Counting (F) \cite{schank2005finding,schank2007}}
\label{alg:forward}
\begin{algorithmic}[1]
\Require{Graph $G = (V, E)$}
\Ensure{Triangle Count $T$}
\State $T \leftarrow 0$
\State $\forall v \in V$
    \State \hspace{8pt} $A(v) \leftarrow \emptyset$
\State $\forall \edge{u, v} \in E$
    \State \hspace{8pt} if $(u<v)$ then
        \State \hspace{16pt} $\forall w \in A(u) \cap A(v)$
            \State \hspace{24pt} $T \leftarrow T+1$
        \State \hspace{16pt} $A(v) \leftarrow A(v) \cup \{ u\}$
        \State return $T$
\end{algorithmic}
\end{algorithm}

The \emph{forward-hashed} algorithm \cite{schank2005finding,schank2007} (also called \emph{compact-forward} \cite{latapy2007practical}) is a variant of the forward algorithm that uses the hashing described previously for the intersections of the $A()$ sets, see Algorithm~\ref{alg:forwardhash}. Low~\emph{et al.} \cite{Low2017} derive a linear algebra method for triangle counting that does not use matrix multiplication. Their algorithm results in the forward-hashed algorithm. 

\begin{algorithm}[htbp]
\footnotesize
\caption{Forward-Hashed Triangle Counting (FH)\cite{schank2005finding,schank2007}}
\label{alg:forwardhash}
\begin{algorithmic}[1]
\Require{Graph $G = (V, E)$}
\Ensure{Triangle Count $T$}
\State $T \leftarrow 0$
\State $\forall v \in V$
    \State \hspace{8pt} $A(v) \leftarrow \emptyset$
\State $\forall \edge{u, v} \in E$
    \State \hspace{8pt} if $(u<v)$ then
        \State \hspace{16pt} $\forall w \in A(u)$
            \State \hspace{24pt} Hash[$w$] $\leftarrow$ true
        \State \hspace{16pt} $\forall w \in A(v)$
            \State \hspace{24pt} if Hash[$w$] then
                \State \hspace{32pt} $T \leftarrow T+1$
        \State \hspace{16pt} $\forall w \in A(u)$
            \State \hspace{24pt} Hash[$w$] $\leftarrow$ false
        \State \hspace{16pt} $A(v) \leftarrow A(v) \cup \{ u\}$
        \State return $T$
\end{algorithmic}
\end{algorithm}

\subsection{Existing Parallel Algorithms}
\label{subsec:related_work_parallel}
Although most of the sequential algorithms tend to run fast on graphs that fit in main memory, the expanding size of graphs, driven by ongoing technology advancements, poses a challenge. To further accelerate, the emergence of parallel version algorithms is inevitable. Alg.~\ref{alg:mergeP}, Alg.~\ref{alg:binaryP}, and Alg.~\ref{alg:hashP} are the parallel versions of the three most common intersection-based triangle counting methods. 

\begin{algorithm}[htbp]
\footnotesize
\caption{Parallel Edge Iterator with Merge Path (EMP)}
\label{alg:mergeP}
\begin{algorithmic}[1]
\Require Graph $G = (V, E)$
\Ensure Triangle Count $T$
\State $T \leftarrow 0$  
\State $\forall(u,v) \in E$ do in parallel
\State \hspace{8pt} $A \leftarrow N(u)$, $B \leftarrow N(v)$ 
\State \hspace{8pt} $x \leftarrow 0$, $y \leftarrow 0$
\State \hspace{8pt} while $x < |A|\ \land \ y < |B|$
\State \hspace{16pt} if $A[x] == B[y]$
\State \hspace{24pt} $T \leftarrow T+1$; 
\State \hspace{24pt} $x \leftarrow x+1$, $y \leftarrow y+1$
\State \hspace{16pt} else
\State \hspace{24pt} if $A[x] < B[y]$
\State \hspace{32pt} $x \leftarrow x+1$
\State \hspace{24pt} else
\State \hspace{32pt} $y \leftarrow y+1$
\State return $T / 6$ 
\end{algorithmic}
\end{algorithm}

\begin{algorithm}[htbp]
\footnotesize
\caption{Parallel Edge Iterator with Binary Search (EBP)}
\label{alg:binaryP}
\begin{algorithmic}[1]
\Require Graph $G = (V, E)$
\Ensure Triangle Count $T$
\State $T \leftarrow 0$  
\State $\forall (u,v) \in E$ do in parallel
\State \hspace{8pt} for $l \in N(u)$ do 
\State \hspace{16pt} $K \leftarrow N(v) $
\State \hspace{16pt} $bottom \leftarrow 0$, $top \leftarrow |K|$
\State \hspace {16pt} while $bottom < top$
\State \hspace {24pt} $mid \leftarrow bottom + (top-bottom)/2$
\State \hspace {24pt} if $K[mid]==l$
\State \hspace {32pt} $T \leftarrow T+1$
\State \hspace {32pt} break
\State \hspace {24pt} elif $K[mid]<l$
\State \hspace {32pt} $bottom \leftarrow mid+1$
\State \hspace {24pt} else
\State \hspace {32pt} $top=mid$
\State return $T / 6$ 
\end{algorithmic}
\end{algorithm}

\begin{algorithm}[htbp]
\footnotesize
\caption{Parallel Edge Iterator with Hash (EHP)}
\label{alg:hashP}
\begin{algorithmic}[1]
\Require Graph $G = (V, E)$
\Ensure Triangle Count $T$
\State $T \leftarrow 0$  
\State $\forall(u,v) \in E$ do in parallel
\State \hspace{8pt} for $w \in N(u)$
\State \hspace{16pt} hash(w) = True
\State \hspace{8pt} for $w \in N(v)$
\State \hspace {16pt} if hash(w)
\State \hspace {24pt} $T \leftarrow T + 1$
\State \hspace {8pt} for $w \in N(u)$
\State \hspace{16pt} hash(w) = False
\State return $T / 6$ 
\end{algorithmic}
\end{algorithm}

Besides the intersection-based methods, there are several optimized parallel algorithms in the literature. Shun \emph{et al.} \cite{shun2015multicore} give a multi-core parallel algorithm for shared memory machines. The algorithm has two steps: in the first step each vertex is ranked based on degree and a ranked adjacency list of each vertex is generated, which contains only higher-ranked vertices than the current vertex; the second step counts triangles from the ranked adjacency list for each vertex using merge-path or hash. Parimalarangan \emph{et al.} \cite{parimalarangan2017fast} present variations of triangle counting algorithms and how they relate to performance in shared-memory platforms. TriCore \cite{hu2018tricore} partitions the graph held in a compressed-sparse row (CSR) data structure for multiple GPUs and uses stream buffers to load edge lists from CPU memory to GPU memory on the fly and then uses binary search to find the intersection. Hu \emph{et~al.} \cite{hu2021accelerating} employ a ``copy-synchronize-search'' pattern to improve the parallel threads efficiency of GPU and mix the computing and memory-intensive workloads together to improve the resource efficiency. Zeng \emph{et~al.} \cite{zeng2022htc} present a triangle counting algorithm that adaptively selects vertex-parallel and edge-parallel paradigm.

\section{Cover-Edge Based Triangle Counting Algorithms}
\label{sec:algorithms}

\subsection{Cover-Edge Set}\label{subsec:cover-edge_sets}

\begin{definition}[Cover-Edge, Cover-Edge Set and Covering Ratio]\label{def:coveredge}
For any edge $e$ of a triangle $\Delta$ in graph $G$, $e$ is referred to as a cover-edge of $\Delta$. For a given graph $G$, an edge set $S \subseteq E$ is called a cover-edge set if it contains at least one cover-edge for every triangle in $G$. $c=|S|/|E|$ is called the covering ratio.
\end{definition}

Based on the given definition, it is evident that the entire edge set $E$ can serve as a cover-edge set $S$ for graph $G$. However, our proposed method aims to efficiently count all triangles using a smaller subset of edges instead of $E$. Thus, the primary challenge lies in generating a compact cover-edge set, which forms the initial problem to be addressed in our approach. Our goal is to identify a cover-edge set with the smallest $c$. In this paper, we propose using breadth-first search (BFS) to generate a compact cover-edge set.

\begin{definition}[BFS-Edge]\label{def:bfsedges}
Let $r$ be the root vertex of an undirected graph $G$. The level $L(v)$ of a vertex $v$ is defined as the shortest distance from $r$ to $v$ obtained through a breadth-first search (BFS). From the BFS, we classify the edges into three types:
\begin{itemize}
\item \emph{Tree-Edges}: These edges belong to the BFS tree.
\item \emph{Strut-Edges}: These are non-tree edges with endpoints on two adjacent levels in the BFS traversal.
\item \emph{Horizontal-Edges}: These are non-tree edges with endpoints on the same level in the BFS traversal.
\end{itemize}
\end{definition}

\begin{figure}
    \centering
    \includegraphics[width=0.35\textwidth]{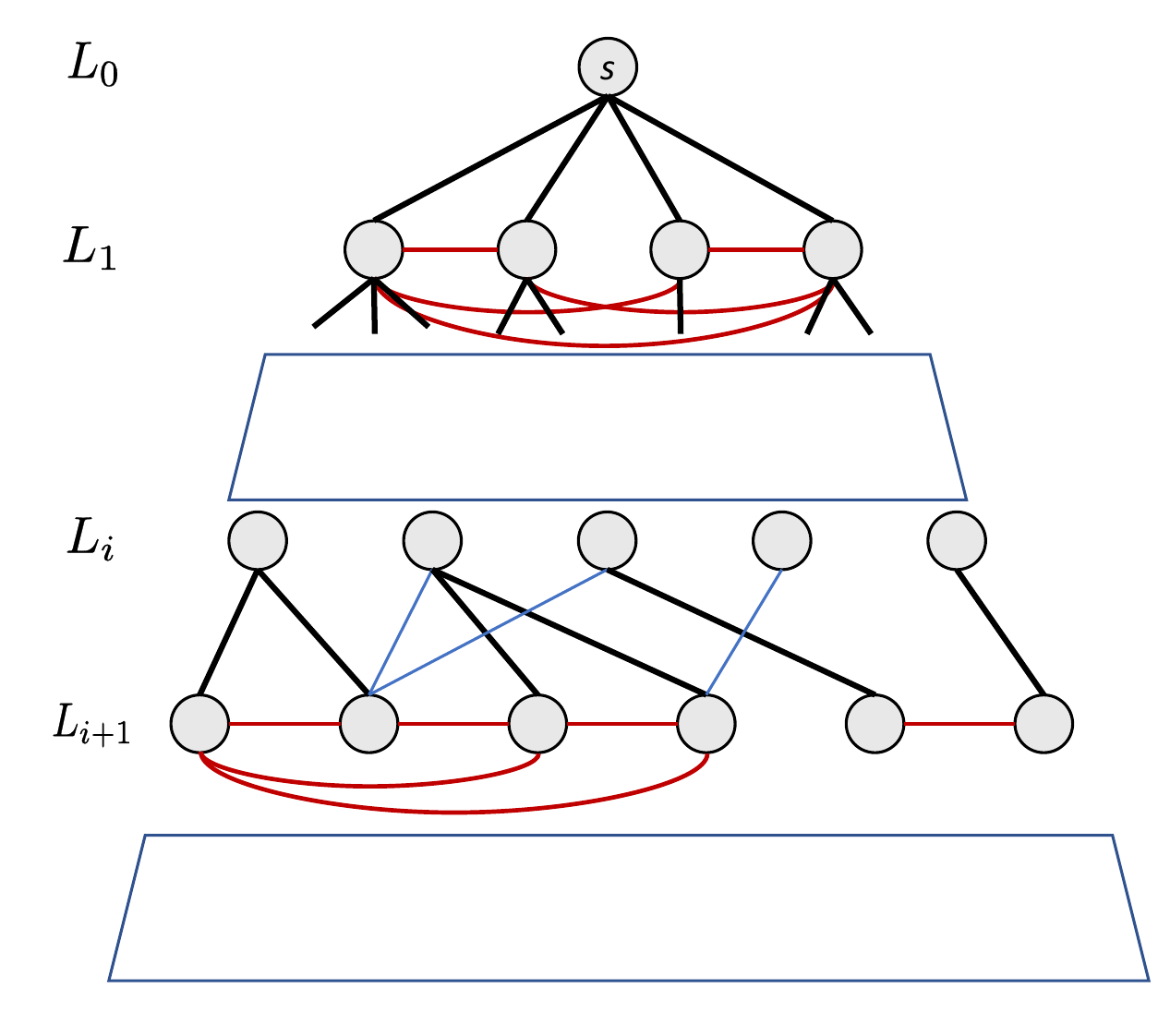}
    \caption{An example to mark different edges based on a BFS spanning tree. The tree-edges are black, strut-edges are blue, and horizontal-edges are red.}
    \label{fig:tree}
\end{figure}

Fig.~\ref{fig:tree} gives an example of these different edge types. 

\begin{lemma}
Each triangle $\{u, v, w\}$ in a graph contains at least one horizontal-edge in an arbitrarily rooted BFS tree.
\label{lem:H-E}
\end{lemma}

\begin{proof}
(Proof by contradiction)
A triangle is a path of length 3 that starts and ends at the same vertex. Suppose there are no horizontal-edges in the triangle. In that case, every edge in the path (i.e., a tree-edge or strut-edge) either increases or decreases the level by one.

Since the path must end on the same level as the starting vertex, the number of edges in the path that decrease the level must be equal to the number of edges that increase the level. Consequently, the length of the path must be even to maintain level parity. However, this contradicts the fact that a triangle has an odd path length of 3.

Therefore, we conclude that there must be at least one horizontal-edge in every triangle.
\end{proof}

\begin{theorem}[Cover-Edge Set Generation]
All horizontal-edges in an arbitrarily rooted BFS tree form a valid cover-edge set.
\label{theo:Gen-CES}
\end{theorem}

\begin{proof}
According to Definition \ref{def:coveredge}, for any triangle $\Delta$ in graph $G$, we can always find at least one horizontal-edge that serves as a cover-edge for $\Delta$. Thus, the set of all horizontal-edges constitutes a cover-edge set.
\end{proof}

Therefore, we can construct a cover-edge set, denoted as \emph{BFS-CES}, by selecting all the horizontal-edges obtained during a breadth-first search (\emph{BFS}). It is evident that \emph{BFS-CES} is a subset of $E$ and is typically much smaller than the complete edge set $E$.  

\subsection{Cover-Edge Triangle Counting: \emph{CETC}}
\label{subsec:CETC}
In this subsection, we provide a comprehensive description of the algorithm to identify all triangles using a cover-edge set generated through a breadth-first search.

\begin{lemma}
Each triangle $\{u, v, w\}$ must contain either one or three horizontal-edges.
\label{lem:HE2}
\end{lemma}

\begin{proof}
By referring to the proof of Lemma~\ref{lem:H-E}, we know that the path corresponding to the triangle's three edges consists of an even number of tree-edges and strut-edges. This implies that there can be either 0 or 2 tree- or strut-edges within each triangle.

In the case where there are 0 tree- or strut-edges, all three edges of the triangle must be horizontal-edges. This is because the absence of tree- or strut-edges implies that the entire path is composed of horizontal-edges.

In the case where there are 2 tree- or strut-edges, the triangle contains exactly one horizontal-edge. This is because having two tree- or strut-edges in the path means that there is one horizontal-edge connecting the remaining two vertices.

Therefore, we conclude that each triangle $\{u, v, w\}$ must contain either one or three horizontal-edges.
\end{proof}

Our sequential triangle counting approach (\emph{CETC-Seq}), described in Alg.~\ref{alg:CETC}, efficiently counts triangles using a cover-edge set. In line \ref{l:init}, we initialize the counter $T$ to 0, which will store the total number of triangles. To generate the cover-edge set, we perform a breadth-first search (BFS) starting from any unvisited vertex, identifying the 
level $L(v)$ of each vertex $v$ in its respective component, as shown in lines \ref{l:bfs1} to \ref{l:bfs2}. In lines \ref{l:edge1} to \ref{l:edge2} the algorithm iterates over each edge, selecting the cover-set of horizontal edges $\edge{u,v}$ in a direction-oriented fashion in line~\ref{l:horiz}.
For each vertex $w$ in the intersection of $u$ and $v$'s neighborhoods (line~\ref{l:intersect}), we check the following two conditions to determine if $(u, v, w)$ is a unique triangle to be counted (line~\ref{l:logic}).
If $L(u) \neq L(w)$  then the edge $\edge{u, v}$ is the only horizontal-edge in the triangle $(u, v, w)$. If $L(u) \equiv L(w)$, then the edge $\edge{u,v}$ is one of three horizontal-edges in the triangle $(u, v, w)$. To ensure uniqueness, the algorithm then checks the added constraint that $v<w$. If the constraints are satisfied, we increment the triangle counter $T$ in line~\ref{l:edge2}.

This approach effectively counts the triangles in the graph while avoiding redundant counting.

\begin{algorithm}[htbp]
\footnotesize
\caption{CETC: Cover-Edge Triangle Counting (\emph{CETC-Seq}) }
\label{alg:CETC}
\begin{algorithmic}[1]
\Require{Graph $G = (V, E)$}
\Ensure{Triangle Count $T$}
\State $T \leftarrow 0$  \label{l:init}
\State $\forall v \in V$ \label{l:bfs1}
\State \hspace{8pt} if $v$ unvisited, then BFS($G$, $v$) \label{l:bfs2}
\State $\forall \edge{u, v} \in E$ \label{l:edge1}
\State \hspace{8pt} if $(L(u) \equiv L(v)) \land (u < v)$ \label{l:horiz} \Comment{$\edge{u,v}$ is horizontal}
\State \hspace{16pt} $\forall w \in N(u) \cap N(v)$ \label{l:intersect}
\State \hspace{24pt} if $(L(u)\neq L(w)) \lor \left( (L(u) \equiv L(w)) \land (v<w) \right)$ then \label{l:logic}
\State \hspace{32pt}  $T\leftarrow T + 1 $     \label{l:edge2} 
\State return $T$
\end{algorithmic}
\end{algorithm}

\begin{theorem}[Correctness]
Alg.~\ref{alg:CETC} can accurately count all triangles in a graph $G$.
\label{theo:correctness}
\end{theorem}
\begin{proof}
Lemma~\ref{lem:HE2} establishes that a triangle in the graph falls into one of two cases: 1) the two endpoint vertices of the horizontal-edge are on the same level while the apex vertex is on a different level, or 2) all three vertices of the triangle are at the same level.

Consider a triangle $\{v_a, v_b, v_c\}$ in $G$. Without loss of generality, assume that $\edge{v_a, v_b}$ is a horizontal-edge, implying $L(v_a) \equiv L(v_b)$. Let $v_c$ be the apex vertex. The two cases can be distinguished as follows:

For the first case, each triangle is uniquely defined by a horizontal-edge and an apex vertex from the common neighbors of the horizontal-edge's endpoint vertices. Whenever Alg.~\ref{alg:CETC} identifies such a triangle $\{v_a,v_b,v_c\}$, it increments the total triangle count $T$ by 1.

In the second case, where all three vertices are at the same level ($L(v_c) \equiv L(v_a) \equiv L(v_b)$), Alg.~\ref{alg:CETC} ensures that $T$ is increased by 1 only when $v_a < v_b < v_c$. This condition ensures that triangle $\{v_a,v_b,v_c\}$ is counted only once, preventing triple-counting and ensuring the correctness of the triangle count.

Hence, Alg.~\ref{alg:CETC} is proven to accurately count all triangles in the graph $G$.
\end{proof}

The time complexity of Alg.~\ref{alg:CETC} can be analyzed as follows. The computation of breadth-first search, including determining the level of each vertex and marking horizontal edges, requires $\mathcal{O}(n+m)$ time.

Since there are at most $\mathcal{O}(m)$ horizontal edges, finding the common neighbors of each horizontal edge individually can be done in $\mathcal{O}(d_{\text{max}})$ time. Here, $d_{\text{max}}$ represents the maximal degree of a vertex in the graph.

Therefore, the overall time complexity of \emph{CETC-Seq} is $ \mathcal{O}(m \cdot d_{\text{max}})$.

\subsection{Variants of \emph{CETC-Seq}}
\subsubsection{CETC Forward Exchanging Triangle Counting Algorithm (\emph{CETC-Seq-FE})}
The overall performance of \emph{CETC-Seq} is closely related to the covering ratio $c$. A higher covering ratio results in fewer reduced edges, which will increase the actual runtime of the algorithm. Therefore, after completing the \emph{BFS}, selecting an appropriate algorithm can be based on $c$. Alg.~\ref{alg:CETC-Seq-FE} presents the variant of \emph{CETC-Seq} that dynamically selects the most suitable approach based on $c$, called \emph{CETC-Seq-FE}. Initially, it calculates $c$ using \emph{BFS} results. If the $c$ value is below a specified threshold (The value of $c$ should be at least less than $(\frac{m-n+1}{m})$. After comparing the performance of Alg.~\ref{alg:CETC} and Alg.~\ref{alg:forward}, we set this threshold to 0.7 in our experiments.), we will continue using Alg.~\ref{alg:CETC}; otherwise, Alg.~\ref{alg:forward} is chosen. Considering the analyses presented in Alg.~\ref{alg:forward} and Alg.~\ref{alg:CETC}, Alg.~\ref{alg:CETC-Seq-FE} maintains a time complexity of \bigO{m^{1.5}}

\begin{algorithm}[htbp]
\footnotesize
\caption{CETC Forward Exchanging (\emph{CETC-Seq-FE})}
\label{alg:CETC-Seq-FE}
\begin{algorithmic}[1]
\Require{Graph $G = (V, E)$}
\Ensure{Triangle Count $T$}
\State $T \leftarrow 0$
\State $\forall v \in V$  
\State \hspace{8pt} if $v$ unvisited, then BFS($G$, $v$)  
\State Calculate $c$ based on the BFS results 
\State If $ (c<threshold)$ 
\State \hspace{8pt} $T \leftarrow \mbox{CETC-Seq}(G)$ \Comment{Alg.~\ref{alg:CETC}}
\State Else 
\State \hspace{8pt} $T \leftarrow \mbox{TC\_forward}(G)$ \Comment{Alg.~\ref{alg:forward}}
\State return $T $ 
\end{algorithmic}
\end{algorithm}

\subsubsection{CETC Split Triangle Counting Algorithm (CETC-Seq-S)}

\begin{algorithm}[htbp]
\footnotesize
\caption{CETC Split Triangle Counting (\emph{CETC-Seq-S})}
\label{alg:CETC-Seq-S}
\begin{algorithmic}[1]
\Require{Graph $G = (V, E)$}
\Ensure{Triangle Count $T$}
\State $T \leftarrow 0$
\State $\forall v \in V$ \label{l:CETC-Seq-S:bfs1}
    \State \hspace{8pt} if $v$ unvisited, then BFS($G$, $v$) \label{l:CETC-Seq-S:bfs2}
\State $\forall \edge{u, v} \in E$ \label{l:CETC-Seq-S:edge1}
    \State \hspace{8pt} if $(L(u) \equiv L(v))$ then \Comment{$\edge{u,v}$ is horizontal}
        \State \hspace{16pt} Add $(u,v)$ to $G_0$
    \State \hspace{8pt} else
        \State \hspace{16pt} Add $(u,v)$ to $G_1$ \label{l:CETC-Seq-S:edge2}
\State $T \leftarrow \mbox{TC}\_{\mbox{forward-hashed}}(G_0)$ \Comment{Alg.~\ref{alg:forwardhash}} \label{l:CETC-Seq-S:forwardhashed}
\State $\forall u \in V_{G_1}$ \label{l:CETC-Seq-S:hash1}
    \State \hspace{8pt} $\forall v \in N_{G_1}(u)$
        \State \hspace{16pt} Hash[$v$] $\leftarrow$ true
    \State \hspace{8pt} $\forall v \in N_{G_0}(u)$
        \State \hspace{16pt} if $(u < v)$ then
            \State \hspace{24pt} $\forall w \in N_{G_1}(v)$
                \State \hspace{32pt} if Hash[$w$] then
                    \State \hspace{40pt} $T \leftarrow T+1$
    \State \hspace{8pt} $\forall v \in N_{G_1}(u)$
        \State \hspace{16pt} Hash[$v$] $\leftarrow$ false \label{l:CETC-Seq-S:hash2}
\State return $T $ 
\end{algorithmic}
\end{algorithm}

Alg.~\ref{alg:CETC-Seq-S} is another variant of \emph{CETC-Seq}, called \emph{CETC-Seq-S}. This variant is similar to cover-edge triangle counting in Alg.~\ref{alg:CETC} and uses BFS to assign a level to each vertex in lines \ref{l:CETC-Seq-S:bfs1} and \ref{l:CETC-Seq-S:bfs2}.
Next in lines~\ref{l:CETC-Seq-S:edge1} to \ref{l:CETC-Seq-S:edge2}, the edges $E$ of the graph are partitioned into two sets $E_0$ -- the horizontal edges where both endpoints are on the same level -- and $E_1$ -- the remaining tree and non-tree edges that span a level. Thus, we now have two graphs, $G_0 = (V, E_0)$ and $G_1 = (V, E_1)$, where $E = E_0 \cup E_1$ and $E_0 \cap E_1 = \emptyset$.  Triangles that are fully in $G_0$ are counted with one method and triangles not fully in $G_0$ are counted with another method.  For $G_0$, the graph with horizontal edges, we count the triangles efficiently using the forward-hashed method (line~\ref{l:CETC-Seq-S:forwardhashed}). For triangles not fully in $G_0$, the algorithm uses the following approach to count these triangles. Using $G_1$, the graph that contains the edges that span levels, we use a hashed intersection approach in lines~\ref{l:CETC-Seq-S:hash1} to \ref{l:CETC-Seq-S:hash2}. As per the cover-edge triangle counting, we need to find the intersections of the adjacency lists from the endpoints of horizontal edges. Thus, we use $G_0$ to select the edges and perform the hash-based intersections from the adjacency lists in graph $G_1$. The proof of correctness for cover-edge triangle counting is given in Section.~\ref{subsec:CETC}. Alg.~\ref{alg:CETC-Seq-S} is a hybrid version of this algorithm, that partitions the edge set, and uses two different methods to count these two types of triangles. The proof of correctness is still valid with these new refinements to the algorithm. The running time of Alg.~\ref{alg:CETC-Seq-S} is the maximum of the running time of forward-hashing and Alg.~\ref{alg:CETC}. Alg.~\ref{alg:CETC-Seq-S} uses hashing for the set intersections. For vertices $u$ and $v$, the cost is $\min(d(u), d(v))$ since the algorithm can check if the neighbors of the lower-degree endpoint are in the hash set of the higher-degree endpoint. Over all $(u,v)$ edges in $E$, these intersections take \bigO{m \cdot a(G)} expected time. Hence, Alg.~\ref{alg:CETC-Seq-S} takes \bigO{m \cdot a(G)} expected time. 

Similar to the forward-hashed method, pre-processing the graph by re-ordering the vertices in decreasing order of degree in \bigTheta{n \log n} time often leads to a faster triangle counting algorithm in practice.

\subsubsection{CETC-Split Recursive Triangle Counting Algorithm (\emph{CETC-Seq-SR})}
\begin{algorithm}[htbp]
\footnotesize
\caption{CETC-Split Recursive Triangle Counting (\emph{CETC-Seq-SR})}
\label{alg:CETC-Seq-SR}
\begin{algorithmic}[1]
\Require{Graph $G = (V, E)$}
\Ensure{Triangle Count $T$}
\State $T \leftarrow 0$
\State $\forall v \in V$ \label{l:CETC-Seq-SR:bfs1}
\State \hspace{8pt} if $v$ unvisited, then BFS($G$, $v$) \label{l:CETC-Seq-SR:bfs2}
\State $\forall \edge{u, v} \in E$ \label{l:CETC-Seq-SR:edge1}
\State \hspace{8pt} if $(L(u) \equiv L(v))$ then \Comment{$\edge{u,v}$ is horizontal}
\State \hspace{16pt} Add $(u,v)$ to $G_0$
\State \hspace{8pt} else
\State \hspace{16pt} Add $(u,v)$ to $G_1$ \label{l:CETC-Seq-SR:edge2}
\State if (size of $G_0>$  threshold) then 
\State \hspace{8pt} $T \leftarrow \mbox{CESR}(G_0)$
\State else
\State \hspace{8pt} $T \leftarrow \mbox{TC}\_{\mbox{forward-hashed}}(G_0)$ \Comment{Alg.~\ref{alg:forwardhash}} \label{l:CETC-Seq-SR:forwardhashed}
\State $\forall u \in V_{G_1}$ \label{l:CETC-Seq-SR:hash1}
\State \hspace{8pt} $\forall v \in N_{G_1}(u)$
\State \hspace{16pt} Hash[$v$] $\leftarrow$ true
\State \hspace{8pt} $\forall v \in N_{G_0}(u)$
        \State \hspace{16pt} if $(u < v)$ then
            \State \hspace{24pt} $\forall w \in N_{G_1}(v)$
                \State \hspace{32pt} if Hash[$w$] then
                    \State \hspace{40pt} $T \leftarrow T+1$
    \State \hspace{8pt} $\forall v \in N_{G_1}(u)$
        \State \hspace{16pt} Hash[$v$] $\leftarrow$ false \label{l:CETC-Seq-SR:hash2}
\State return $T $ 
\end{algorithmic}
\end{algorithm}

The Alg.~\ref{alg:CETC-Seq-SR} is similar to Alg.~\ref{alg:CETC-Seq-S}. The only difference is that for the subgraph $G_0$ consisting of the horizontal edges. If its size is larger than the given threshold value, we will recursively call the algorithm to further reduce the graph size (line 9).  If the size of $G_0$ is not larger than the given threshold value, we will directly call Alg.~\ref{alg:forwardhash} to get the total number of triangles in $G_0$ (line 11). We use the same threshold value of 0.7 in the experiment as outlined in Alg.~\ref{alg:CETC-Seq-FE}. The idea behind the recursive call is that we can quickly count the triangles containing edges across both $G_0$ and $G_1$, and then we can safely remove all the edges in $G_1$ to reduce the graph size. Finally, Alg.~\ref{alg:forwardhash} will focus on a smaller graph whose edges may contain multiple triangles.

\subsection{Parallel CETC Algorithm on Shared-Memory (\emph{CETC-SM})}

In Section~\ref{sec:related}, we introduced three commonly employed intersection-based methods: merge-path, binary search, and hash, alongside their corresponding parallel version as outlined in Alg.~\ref{alg:mergeP}, Alg.~\ref{alg:binaryP}, and Alg.~\ref{alg:hashP}.

The fundamental concept behind the proposed parallel algorithms is to calculate the intersection of neighbor lists of two endpoints of any $\edge{u,v}$ in parallel, which will significantly increase the performance.

\begin{algorithm}[htbp]
\caption{Shared Memory Parallel Cover-Edge Triangle Counting (\emph{CETC-SM})}
\label{alg:parallel}
\begin{algorithmic}[1]
\Require{Graph $G = (V, E)$}
\Ensure{Triangle Count $T$}
\State $c_1, c_2 \leftarrow 0$
\State Run Parallel BFS on $G$ and mark the level.
\State $\forall (u,v) \in E$ do in parallel
\State \hspace{8pt} if $(L(u) \equiv L(v)) \land (u < v)$ \Comment{$\edge{u,v}$ is horizontal}
\State \hspace{16pt} $\forall w \in N(u) \cap N(v)$
\State \hspace{24pt} if ($L(w) \neq L(u)$) then
\State \hspace{32pt} $c_1\leftarrow c_1 + 1$ 
\State \hspace{24pt} else
\State \hspace{32pt} $c_2 \leftarrow c_2 + 1$
\State $T \leftarrow c_1 + c_2/3$
\State return $T$ \Comment{See Alg.~\ref{alg:CETC}}
\end{algorithmic}
\end{algorithm}

Alg.\ref{alg:parallel} demonstrates the parallelization of the Covering-Edge triangle counting algorithm for shared-memory. In the context of the PRAM (Parallel Random Access Machines) model, both parallel \emph{BFS} and parallel sorting have been shown to achieve scalable performance \cite{cormen2022introduction}. For set intersection operations on a single edge, it is imperative that the computation remains well below $m^{0.5}$\cite{schank2007}, particularly when dealing with large input graphs, where $p$ represents the total number of processors. Consequently, the total work, which is $\mathcal{O}(m^{1.5})$, can be evenly distributed among $p$ processors. As a result, \emph{CETC-SM} exhibits a parallel time complexity of $\mathcal{O}(\frac{m^{1.5}}{p})$, ensuring scalability as the number of parallel processors increases.

\subsection{Communication-Efficient Parallel CETC Algorithm on Distributed-Memory (\emph{CETC-DM})}
\label{sec:parallel}
This subsection presents our communication-efficient parallel algorithm for counting triangles in massive graphs on a $p$-processor distributed-memory parallel computer. We will take advantage of the concept of \emph{Cover-Edge Set} to significantly improve the communication performance of our triangle counting method.
Since distributed triangle counting is communication-bound \cite{pearce2017triangle}, this algorithm is expected to improve the overall running time.
The input graph $G$ is stored in a compressed sparse row (CSR) format. The vertices are partitioned non-uniformly to the $p$ processors such that each processor stores approximately $2m/p$ edge endpoints.
This graph input follows the format used by the majority of parallel graph algorithm implementations and benchmarks such as Graph500 and Graph Challenge. 

Our communication-efficient parallel algorithm \emph{CETC-DM} (see Alg.~\ref{algparallel}) is based on the same cover-edge approach proposed in Section \ref{subsec:CETC}. %
The binary operator $\oplus$ used in line~\ref{l:swap} is bitwise exclusive OR (XOR).

\begin{algorithm}[htbp]
\footnotesize
\caption{CETC Communication Efficient Triangle Counting (\emph{CETC-DM})}
\label{algparallel}
\begin{algorithmic}[1]
\Require{Graph $G = (V, E)$}
\Ensure{Triangle Count $T$}
\State Run parallel BFS($G$) and build partial cover-edge set $S_i$ on $p_i$ \label{l:cbfs}
\State For all $p_i, i \in \{0 \ldots p-1\}$ in parallel do: \label{l:bcast1}
    \State \hspace{8pt} $t_i \leftarrow 0$  \label{l:pinit}
    \State \hspace{8pt}  $\forall \edge{u,v} \in S_i$ with $u<v$ on $p_i$ \label{l:sloop1}
        \State \hspace{16pt} $\forall w \in V_i$ such that $w \in N(u), N(v)$ \label{l:select1}
            \State \hspace{24 pt}  if $(L(u)\neq L(w)) \lor \left( (L(u) \equiv L(w)) \land (v<w) \right)$  then \label{l:count1}
                \State \hspace{32 pt} $t_i = t_i + 1$ \label{l:eloop1}
    \State \hspace{8pt} For $j \leftarrow 1$ to $p-1$ do: \label{l:bgraph}
        \State \hspace{16pt} Processors $i$ and $i \oplus j$ swap edge sets $S_i$ and $S_j$. \label{l:swap}
        \State \hspace{16pt}  $\forall \edge{u, v} \in S_j$ with $u<v$ on $p_i$ \label{l:sloop2}
            \State \hspace{24pt} $\forall w \in V_i$ such that $w \in N(u), N(v)$ \label{l:select2}
                 \State \hspace{32 pt}  if $(L(u)\neq L(w)) \lor \left( (L(u) \equiv L(w)) \land (v<w) \right)$ then \label{l:count2}
                    \State \hspace{40 pt} $t_i = t_i + 1$  \label{l:eloop2}
\State $T \leftarrow$ Reduce$(t_{i}, +)$ \label{l:reduce}
\end{algorithmic}
\end{algorithm}

Similar to the sequential \emph{CETC-Seq} algorithm, the cover-edge set $S = \cup_{i=0}^{p-1} S_i$ is determined in line~\ref{l:cbfs} by labeling the horizontal edges from a parallel BFS.

Each processor runs lines~\ref{l:bcast1} to \ref{l:eloop2} in parallel that consists of two main substeps. Local triangles are counted in lines~\ref{l:sloop1} to \ref{l:eloop1} and a total exchange of cover-edges between each pair of processors to count triangles is performed in lines~\ref{l:bgraph} to \ref{l:eloop2}. 
Note at the end of each iteration of the \emph{for} loop, processor $p_i$ can discard the cover-edge set $S_j$.
In lines~\ref{l:select1} and \ref{l:select2}, processor $p_i$ determines for each cover edge $\edge{u, v}$ all the apex vertices $w$ held locally that are adjacent to both $u$ and $v$. The logic for counting triangles in lines~\ref{l:count1} and \ref{l:count2} is similar to Alg.~\ref{alg:CETC} as to only count unique triangles. Finally, a reduction operation in line~\ref{l:reduce} calculates the total number of triangles by accumulating the $p$ triangle counters, i.e., $T = \sum_{i=0}^{p-1} t_i$.

\subsubsection{Cost Analysis}

\paragraph{Space}

In addition to the input graph data structure, an additional bit is needed per edge (for marking a horizontal-edge) and \bigO{\lceil \log D \rceil} bits per vertex to store its level, where $D$ is the diameter of the graph. This is a total of at most $m + n\lceil \log D \rceil$ bits across the $p$ processors. Preserving the graph requires additional \bigO{n+m} space for the graph. 

\paragraph{Compute}

The BFS costs \bigO{(n+m)/p} \cite{cormen2022introduction}. The search corresponding to one cover-edge in a vertex's adjacency list takes at most \bigO{\log(d_{max})} time using binary search, and only \bigO{1} expected time using a hash table. Let $d_i$ be the degree of vertex $v_i$ where $0\le i<n$. Searching $km$ edges in all vertices' adjacency lists takes $\mathcal{O}(km\sum_{i=0}^{n-1} \log(d_i)) = \mathcal{O}(km \log(\Pi_{i=0}^{n-1}d_i))$ time. Since $\sum_{i=0}^{n-1}d_i=2m$, we know that $\log(\Pi_{i=0}^{n-1}d_i)$ reaches its maximum value when $d_i=2m/n$ for $0\le i<n$. Thus, $\mathcal{O}(km \log(\Pi_{i=0}^{n-1}d_i)) \le \mathcal{O}(km \log((2m/n)^n)) \le \mathcal{O}(kmn \log(2n^2/n)) = \mathcal{O}(mn \log(n))$.

\paragraph{Total Communication}

In our analysis of communication cost for BFS, we measure the total communication volume independent of the number of processors. Thus, this is a conservative overestimate of communication since a fraction (e.g., $1/p$) of accesses will be on the same compute node versus message traffic between nodes. At the same time, we do not consider the savings from overlapping with the computation cost.

The cost of the breadth-first search is $m$ edge traversals with $\lceil \log D \rceil + 3 \lceil \log n \rceil$ bits communicated per edge traversal for the level information, pair of vertex ids, and vertex degree, yielding $m \cdot (\lceil \log D \rceil + 3 \lceil \log n \rceil)$ bits for the BFS.
Transferring $km$ horizontal-edges requires $kmp \lceil \log n \rceil$ bits, where $p$ is the number of processors.
The final reduction to find the total number of triangles requires $(p-1)\lceil \log n \rceil$ bits.

Hence, the total communication volume is 
$m \cdot (\lceil \log D \rceil + 3 \lceil \log n \rceil) + 
kmp \lceil \log n \rceil +
(p-1) \lceil \log n \rceil
=
m \cdot (\lceil \log D \rceil + (kp + 3) \lceil\log n \rceil) + (p -1) \lceil \log n \rceil$
bits.  Hence, since the word size is $\bigTheta{\log n}$ and $D \leq n$, 
the communication of \emph{CETC-DM} is \bigO{pm} words.

\section{Open Source Evaluation Framework}
In the preceding sections, we presented all sequential and shared-memory algorithms from literature known to the authors plus our novel approaches. In this section, we introduce our open-source framework designed to integrate comprehensive triangle counting implementations.

There lacks a unified framework encompassing all implementations, which is important for researchers to conduct performance comparisons among existing algorithms and to assess their efficacy against newly proposed methods. Consequently, we have developed a comprehensive open-source framework to solve this problem. This framework is designed to ensure a thorough evaluation of triangle counting algorithms. It includes implementations of 22 sequential methods and 11 parallel methods on shared-memory as complete a set of what is found in the literature.

Each triangle counting routine has a single argument -- a pointer to the graph in a compressed sparse row (CSR) format. The input is treated as read-only. Each algorithm is charged the full cost if the implementation needs auxiliary arrays, pre-processing steps, or additional data structures. Each implementation must manage memory and not contain any memory leaks -- hence, any dynamically allocated memory must be freed before returning the result.  

The output from each implementation is an integer with the number of triangles found.  Each algorithm is run ten times, and the mean running time is reported.  To reduce variance for random graphs, the same graph instance is used for all of the experiments.  For sequential algorithms, the source code is sequential C code without any explicit parallelization. For parallel algorithms, we use OpenMP to parallelize the C code. The same coding style and effort were used for each implementation. 

Here, we list algorithms subjected to the experiments given in the next section, including both established methods and the newly proposed algorithms.  Algorithms then end with $P$ indicate that we have also developed parallel versions.

\begin{description}[style=unboxed,leftmargin=0cm]
\item[W/WP]: Wedge-checking/Parallel version
\item[WD/WDP]: Wedge-checking(direction-oriented)/Parallel version
\item[EM/EMP]: Edge Iterator with MergePath for set intersection/Parallel version
\item[EMD/EMDP]: Edge Iterator with MergePath for set intersection (direction-oriented)/Parallel version
\item[EB/EBP]: Edge Iterator with BinarySearch for set intersection/Parallel version
\item[EBD/EBDP]: Edge Iterator with BinarySearch for set intersection (direction-oriented)/Parallel version
\item[ET/ETP]: Edge Iterator with partitioning for set intersection/Parallel version
\item[ETD/ETDP]: Edge Iterator with partitioning for set intersection (direction-oriented)/Parallel version
\item[EH/EHP]: Edge Iterator with Hashing for set intersection/Parallel version
\item[EHD/EHDP]: Edge Iterator with Hashing for set intersection (direction-oriented)/Parallel version
\item[F]: Forward
\item[FH]: Forward with Hashing
\item[FHD]: Forward with Hashing and degree-ordering
\item[TS]: Tri\_simple (Davis \cite{davis2018HPEC})
\item[LA]: Linear Algebra (CMU \cite{Low2017})
\item[IR]: Treelist from Itai-Rodeh \cite{itai1978finding}
\item[CETC-Seq/CETC-SM]: Cover Edge Triangle Counting (Bader, \cite{10363465})/Parallel version on shared-memory
\item[CETC-Seq-D]: Cover Edge Triangle Counting with degree-ordering (Bader, \cite{10363465})
\item[CETC-Seq-FE]: Cover Edge Forward Exchanging Triangle Counting 
\item[CETC-Seq-S]: Cover Edge Split Triangle Counting (Bader, \cite{10363539})
\item[CETC-Seq-SD]: Cover Edge Split Triangle Counting with degree-ordering (Bader, \cite{10363539})
\item[CETC-Seq-SR]: Cover Edge-Split Recursive Triangle Counting
\end{description}

\section{Experimental Results}
\subsection{Platform Configuration}

We use the Intel Development Cloud for benchmarking our results on a GNU/Linux node. The compiler is Intel(R) oneAPI DPC++/C++ Compiler 2023.1.0 (2023.1.0.20230320) and `\texttt{-O2}` is used as a compiler optimization flag. we use a recently launched Intel Xeon processor (Sapphire Rapids launched Q1'23) with DDR5 memory for both sequential and parallel implementations. The node is a dedicated 2.00 GHz 56-core (112 thread) Intel(R) Xeon(R) Platinum 8480+ processor (formerly known as Sapphire Rapids) with 105M cache and 1024GB of DDR5 RAM. 

\subsection{Data Sets}

We employ a diverse collection of graphs. The real-world datasets are from SNAP. For the synthetic graphs, we use large Graph500 RMAT graphs \cite{chakrabarti2004r} with parameters $a = 0.57$, $b = 0.19$, $c = 0.19$, and $d = 0.05$, similar to the IARPA AGILE benchmark graphs.
An overview of all 24 graphs in our dataset is presented in Table \ref{tab:dataset}. 

The values of $c$ exhibit substantial variation across different graphs, ranging from 0.90 to 0.14. Smaller $c$ values signify a higher potential for avoiding fruitless searches, thereby enhancing the efficiency of our approach.

\begin{table}
\scriptsize
\centering
\caption{Data Sets for the Experiments}
\label{tab:dataset}
\begin{tabular}{|c|c|c|c|c|c|}
\hline
Graph   Name   & Graph ID & n       & m       & \# triangles & $c$ (\%) \\ \hline
RMAT 6         & 1        & 64      & 1024    & 9100         & 93.8     \\ \hline
RMAT 7         & 2        & 128     & 2048    & 18855        & 90.9     \\ \hline
RMAT 8         & 3        & 256     & 4096    & 39602        & 87.6     \\ \hline
RMAT 9         & 4        & 512     & 8192    & 86470        & 87.2     \\ \hline
RMAT 10        & 5        & 1024    & 16384   & 187855       & 82.8     \\ \hline
RMAT 11        & 6        & 2048    & 32768   & 408876       & 81.1     \\ \hline
RMAT 12        & 7        & 4096    & 65536   & 896224       & 77.5     \\ \hline
RMAT 13        & 8        & 8192    & 131072  & 1988410      & 74.9     \\ \hline
RMAT 14        & 9        & 16384   & 262144  & 4355418      & 70.5     \\ \hline
RMAT 15        & 10       & 32768   & 524288  & 9576800      & 68.4     \\ \hline
RMAT 16        & 11       & 65536   & 1048576 & 21133772     & 65.5     \\ \hline
RMAT 17        & 12       & 131072  & 2097152 & 46439638     & 62.8     \\ \hline
karate         & 13       & 34      & 78      & 45           & 35.9     \\ \hline
amazon0302     & 14       & 262111  & 899792  & 717719       & 44.2     \\ \hline
amazon0312     & 15       & 400727  & 2349869 & 3686467      & 52.4     \\ \hline
amazon0505     & 16       & 410236  & 2439437 & 3951063      & 52.7     \\ \hline
amazon0601     & 17       & 403394  & 2443408 & 3986507      & 52.8     \\ \hline
loc-Brightkite & 18       & 58228   & 214078  & 494728       & 43.2     \\ \hline
loc-Gowalla    & 19       & 196591  & 950327  & 2273138      & 50.8     \\ \hline
roadNet-CA     & 20       & 1971281 & 2766607 & 120676       & 14.5     \\ \hline
roadNet-PA     & 21       & 1090920 & 1541898 & 67150        & 14.6     \\ \hline
roadNet-TX     & 22       & 1393383 & 1921660 & 82869        & 14       \\ \hline
soc-Epinions1  & 23       & 75888   & 405740  & 1624481      & 53.3     \\ \hline
wiki-Vote      & 24       & 8297    & 100762  & 608389       & 54.3     \\ \hline
\end{tabular}
\end{table}

\subsection{Results and Analysis of Sequential Algorithms}
The execution times of the sequential algorithms (in seconds) are presented in Table.~\ref{t:results:sequential}.

\newsavebox{\boxSequential}
\begin{lrbox}{\boxSequential}
\begin{tabular}{lrrrrrrrrrrr}
Graph          & W         & WD        & EM       & EMD      & EB       & EBD      & ET       & ETD      & EH       & EHD      & F  \\ \hline
RMAT 6  & 0.0023	& 0.000435	& 0.000608	& 0.000301	& 0.001534	& 0.000752	& 0.001722	& 0.000896	& 0.000118	& 0.000055	& 0.000055 \\
RMAT 7  & 0.005482	& 0.001603	& 0.001973	& 0.000992	& 0.004567	& 0.002257	& 0.005145	& 0.002765	& 0.000354	& 0.000166	& 0.000219 \\
RMAT 8  & 0.016084	& 0.004622	& 0.005455	& 0.002726	& 0.012873	& 0.006365	& 0.01379	& 0.007551	& 0.000866	& 0.000429	& 0.000642 \\
RMAT 9  & 0.04884	& 0.014317	& 0.014643	& 0.007337	& 0.031095	& 0.014825	& 0.018992	& 0.010409	& 0.001124	& 0.000567	& 0.000912 \\
RMAT 10 & 0.081409	& 0.023978	& 0.020342	& 0.01015	& 0.046255	& 0.023038	& 0.049716	& 0.02734	& 0.00289	& 0.001477	& 0.002433 \\
RMAT 11 & 0.260367	& 0.077086	& 0.053735	& 0.026881	& 0.109204	& 0.054332	& 0.128502	& 0.071314	& 0.006977	& 0.00352	& 0.006196 \\
RMAT 12 & 0.896607	& 0.262051	& 0.141176	& 0.070548	& 0.293548	& 0.146626	& 0.331856	& 0.185648	& 0.017128	& 0.008613	& 0.015863 \\
RMAT 13 & 2.975701	& 0.876912	& 0.372609	& 0.186514	& 0.73989	& 0.369528	& 0.849023	& 0.476537	& 0.044211	& 0.022125	& 0.040797 \\
RMAT 14 & 10.520327	& 3.108799	& 0.987748	& 0.492118	& 1.829937	& 0.914373	& 2.192774	& 1.241524	& 0.114199	& 0.056226	& 0.104752 \\
RMAT 15 & 35.785918	& 10.461789	& 2.626837	& 1.307338	& 4.834125	& 2.397788	& 5.607495	& 3.185066	& 0.31823	& 0.152634	& 0.27252 \\
RMAt 16 & 122.100925	& 35.690483	& 6.931398	& 3.452957	& 12.020692	& 5.942763	& 14.38633	& 8.202219	& 1.072639	& 0.51392	& 0.714004 \\
RMAt 17 & 426.945522	& 124.153096	& 18.328512	& 9.153039	& 31.596123	& 15.577652	& 37.014029	& 21.238411	& 3.249189	& 1.582771	& 1.865376 \\
karate & 0.000015	& 0.000007	& 0.000012	& 0.000006	& 0.000019	& 0.000009	& 0.000033	& 0.000014	& 0.000009	& 0.000005	& 0.000004 \\
amazon0302 & 0.298977	& 0.052727	& 0.143474	& 0.066293	& 0.193383	& 0.090165	& 0.295645	& 0.152895	& 0.06663	& 0.03324	& 0.024458 \\
amazon0312 & 1.293403	& 0.410326	& 0.720808	& 0.352071	& 1.007608	& 0.474253	& 1.587602	& 0.895909	& 0.286009	& 0.135694	& 0.108818 \\
amazon0505 & 1.4014	& 0.458928	& 0.75572	& 0.370004	& 1.071849	& 0.505645	& 1.687541	& 0.950925	& 0.296272	& 0.140679	& 0.114741 \\
amazon0601 & 1.400537	& 0.458847	& 0.762124	& 0.372543	& 1.085706	& 0.511053	& 1.694514	& 0.952832	& 0.300645	& 0.142758	& 0.118401 \\
loc-Brightkit & 0.473063	& 0.158358	& 0.115803	& 0.058309	& 0.1789	& 0.089439	& 0.262685	& 0.153408	& 0.025772	& 0.013417	& 0.011734 \\
loc-Gowalla & 9.018113	& 4.425076	& 2.083049	& 1.03675	& 1.247894	& 0.610038	& 2.850531	& 2.066343	& 0.354619	& 0.173588	& 0.085867 \\
roadNet-CA & 0.089097	& 0.032726	& 0.102766	& 0.064406	& 0.106291	& 0.06895	& 0.168963	& 0.096467	& 0.073492	& 0.051154	& 0.035571 \\
roadNet-PA & 0.074498	& 0.035727	& 0.07553	& 0.036011	& 0.060924	& 0.039204	& 0.097066	& 0.05476	& 0.041419	& 0.028734	& 0.019805 \\
roadNet-TX & 0.088466	& 0.023446	& 0.070821	& 0.044442	& 0.071961	& 0.047007	& 0.117474	& 0.066475	& 0.050618	& 0.035551	& 0.024401 \\
soc-Epinions1 & 4.892297	& 1.853538	& 0.615327	& 0.306373	& 1.144934	& 0.569414	& 1.540422	& 0.876375	& 0.09952	& 0.047979	& 0.062959 \\
wiki-Vote & 0.830841	& 0.210642	& 0.12436	& 0.062224	& 0.290031	& 0.145106	& 0.355414	& 0.190386	& 0.019756	& 0.009999	& 0.01816 \\[12pt]

Graph           & FH        & FHD       & TS        & LA       & IR         & CETC-Seq        & CETC-Seq-D       & CETC-Seq-FE        & CETC-Seq-S        & CETC-Seq-SD        & CETC-Seq-SR   \\ \hline
RMAT 6 & 0.000028	& 0.00003	& 0.000045	& 0.000049	& 0.001057	& 0.000096	& 0.00009	& 0.000036	& 0.000042	& 0.000041	& 0.000034 \\
RMAT 7 & 0.000076	& 0.000079	& 0.000118	& 0.0002	& 0.00343	& 0.00045	& 0.000337	& 0.000083	& 0.000117	& 0.000105	& 0.000112 \\
RMAT 8 & 0.00021	& 0.000213	& 0.000341	& 0.000601	& 0.010451	& 0.00118	& 0.001008	& 0.000242	& 0.000316	& 0.000326	& 0.000292 \\
RMAT 9 & 0.000279	& 0.000284	& 0.000474	& 0.000844	& 0.017542	& 0.001573	& 0.001312	& 0.000309	& 0.000402	& 0.000415	& 0.000393 \\
RMAT 10 & 0.000633	& 0.000623	& 0.001238	& 0.002231	& 0.056125	& 0.003793	& 0.00329	& 0.000721	& 0.000929	& 0.000907	& 0.000905 \\
RMAT 11 & 0.001469	& 0.001386	& 0.003257	& 0.005701	& 0.175472	& 0.008988	& 0.007996	& 0.001635	& 0.001938	& 0.002022	& 0.001902 \\
RMAT 12 & 0.0034	& 0.003114	& 0.008204	& 0.014481	& 0.541587	& 0.020948	& 0.018846	& 0.003683	& 0.004291	& 0.004306	& 0.004719 \\
RMAT 13 & 0.007899	& 0.006939	& 0.021101	& 0.036664	& 1.762489	& 0.048973	& 0.043984	& 0.008508	& 0.009416	& 0.009234	& 0.010386 \\
RMAT 14 & 0.018383	& 0.015386	& 0.05573	& 0.091965	& 5.926348	& 0.112335	& 0.100385	& 0.019555	& 0.02067	& 0.019464	& 0.020786 \\
RMAT 15 & 0.045298	& 0.038136	& 0.212061	& 0.23366	& 19.623668	& 0.266065	& 0.233845	& 0.982818	& 0.047584	& 0.043019	& 0.046379 \\
RMAT 16 & 0.120033	& 0.095672	& 0.502867	& 0.595802	& 63.575078	& 0.630521	& 0.539548	& 2.500813	& 0.117718	& 0.098857	& 0.111664 \\
RMAT 17 & 0.326685	& 0.25219	& 1.509844	& 1.513653	& 209.558816	& 1.491119	& 1.245597	& 6.385706	& 0.325907	& 0.241169	& 0.284544 \\
karate & 0.000004	& 0.000008	& 0.000005	& 0.000002	& 0.000093	& 0.000006	& 0.00001	& 0.000006	& 0.000009	& 0.000011	& 0.000009 \\
amazon0302 & 0.01929	& 0.03971	& 0.038784	& 0.0228	& 0.375845	& 0.04485	& 0.066	& 0.05224	& 0.044427	& 0.061649	& 0.064679 \\
amazon0312 & 0.067849	& 0.120262	& 0.164739	& 0.100246	& 2.223462	& 0.145356	& 0.195258	& 0.216578	& 0.121874	& 0.165797	& 0.207789 \\
amazon0505 & 0.070928	& 0.125915	& 0.169891	& 0.10574	& 2.194572	& 0.151568	& 0.204355	& 0.225309	& 0.130232	& 0.174099	& 0.200114 \\
amazon0601 & 0.072972	& 0.127608	& 0.173651	& 0.107426	& 2.112184	& 0.153916	& 0.207247	& 0.229265	& 0.132083	& 0.176369	& 0.202132 \\
loc-Brightkit & 0.006487	& 0.00982	& 0.013752	& 0.011592	& 0.58484	& 0.011902	& 0.015295	& 0.028136	& 0.008623	& 0.012866	& 0.011192 \\
loc-Gowalla & 0.038467	& 0.054073	& 0.188503	& 0.079424	& 6.433136	& 0.074779	& 0.088801	& 0.276325	& 0.045462	& 0.063348	& 0.061056 \\
roadNet-CA & 0.03906	& 0.170009	& 0.05232	& 0.038795	& 0.656398	& 0.083151	& 0.209056	& 0.100744	& 0.125747	& 0.235934	& 0.141022 \\
roadNet-PA & 0.021824	& 0.083848	& 0.029642	& 0.021815	& 0.366677	& 0.044418	& 0.109172	& 0.054114	& 0.061127	& 0.123917	& 0.068803 \\
roadNet-TX & 0.027021	& 0.106699	& 0.03607	& 0.026994	& 0.506742	& 0.054871	& 0.138155	& 0.067203	& 0.075937	& 0.163585	& 0.08579 \\
soc-Epinions1 & 0.019544	& 0.022428	& 0.051883	& 0.063396	& 4.939747	& 0.043198	& 0.041572	& 0.168057	& 0.021562	& 0.023076	& 0.026793 \\
wiki-Vote & 0.004964	& 0.00498	& 0.009605	& 0.019463	& 0.545068	& 0.013	& 0.014294	& 0.034383	& 0.005118	& 0.00535	& 0.005741 
\end{tabular}
\end{lrbox}

\begin{table*}
\scriptsize
\centering
\caption{Execution time (in seconds) for sequential algorithms.}
\scalebox{0.56}{\usebox{\boxSequential}}
\label{t:results:sequential}
\end{table*}

\subsubsection{Effect of Direction-Oriented on Sequential Algorithms} 
\label{subsec:DOExp}

The \emph{DO} performance optimization is a pivotal strategy in triangle counting, designed to mitigate redundant calculations. In this section, we explore five distinct duplicate counting algorithms, each accompanied by its corresponding \emph{DO} variant. The results presented in Fig. \ref{fig:DO} vividly demonstrate the speedup achieved by the \emph{DO} counterparts compared to their duplicate counting versions.

Evidently, across all scenarios, the majority of \emph{DO} algorithms yield a speedup of at least two-fold. Particularly, the \emph{WD} algorithm stands out with a higher average speedup of 3.637, surpassing the performance gains of other algorithms. \emph{EBD} exhibits a speedup of $2.015 \times$, closely followed by \emph{EMD} at $2.005 \times$, \emph{EHD} at $1.965 \times$, and \emph{ETD} at $1.784 \times$.

\emph{DO} optimization primarily constitutes an algorithmic enhancement, resulting in a reduction in the overall number of operations. So, for any graph, it can improve the performance and our experimental results also confirm its efficiency. However, the practical performance gains can be impacted by various factors, including memory access patterns and cache utilization. Our comprehensive experiments, conducted on diverse graphs using a range of algorithms, underscore the substantial performance enhancements achievable through \emph{DO} optimization.

In summary, \emph{DO} optimization is efficient for eliminating duplicate triangle counting and significantly improving overall performance.

\begin{figure}
    \centering
    \includegraphics[width=0.895\textwidth]{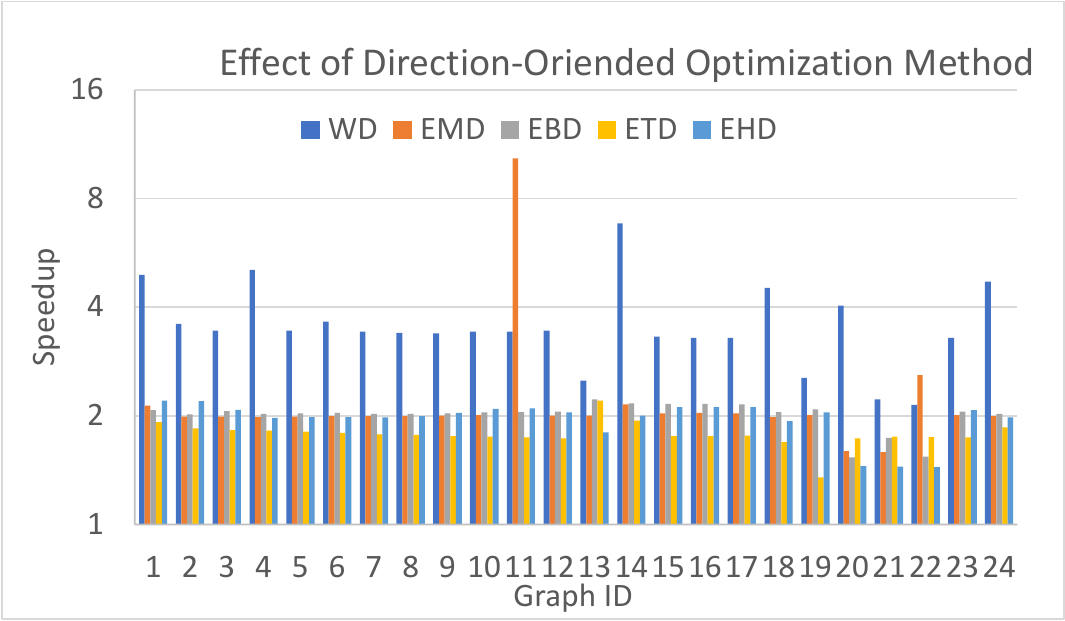}
    \caption{The speedups of direction-oriented optimization compared with the duplicate counting counterparts.}
    \label{fig:DO}
\end{figure}

\subsubsection{Effect of Hash Method on Sequential Algorithms} 
\label{subsec:hashExp}

Similar to the \emph{DO} optimization, the hash-based optimization proves highly efficient in most scenarios. In Fig. \ref{fig:Hash}, we illustrate the speedup achieved by \emph{Hash} methods compared to non-hash implementations. The first comparison showcases the speedup of the \emph{Hash} set intersection (\emph{EH}) compared with the non-hashed method (\emph{EM}), while the second presents the speedup of (\emph{FH}) compared with (\emph{F}).

The average speedup of \emph{EH} is $5.4 \times$, and for \emph{FH}, it is $3.0 \times$, underscoring the effectiveness of the hash-based optimization. Notably, the results reveal that, for more efficient algorithms, like \emph{F}, the speedup is slightly lower than that of less efficient algorithms, such as \emph{EM}.

However, we observe several exceptions. For roadNet-CA (Graph ID=20), roadNet-PA (Graph ID=21), and roadNet-TX (Graph ID=22), the \emph{Hash} algorithm \emph{FH} performs worse than the non-hashed algorithm \emph{F}. This is attributed to the unique topologies of these graphs, characterized by relatively long diameters and very few neighbors for each vertex. As the intersection sets are relatively small, the \emph{MergePath} operation on small sets proves more efficient than the \emph{Hash} method, given the relatively high hash table overhead for very small sets. Therefore,  the \emph{Hash} optimization method remains efficient but not for some special topology and diameter graphs, as the hash table overhead may not compensate for small intersection sets.

\begin{figure}
    \centering
    \includegraphics[width=0.895\textwidth]{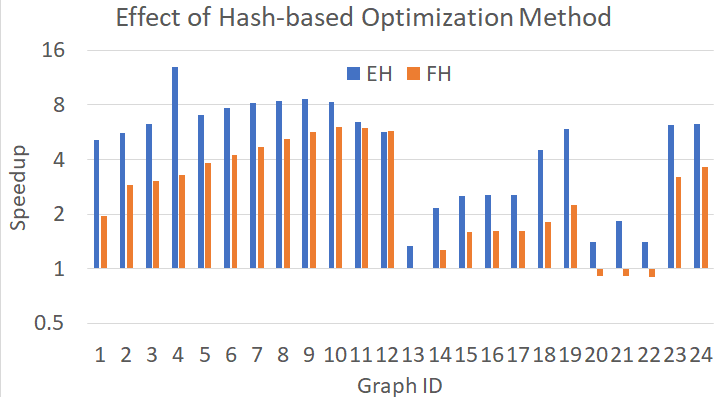}
    \caption{The speedups of hash-based optimization compared with the \emph{MergePath} method.}
    \label{fig:Hash}
\end{figure}

\subsubsection{Effect of Forward Algorithm and Its Variants} 
\label{subsec:SRExp} 

Our experimental results underscore the effectiveness of \emph{Forward} algorithm and its variants as robust algorithms for enhancing the performance of triangle counting. In Fig. \ref{fig:SIS}, we present the speedup achieved by three algorithms—namely, the forward algorithm (\emph{F}), the hashed forward algorithm (\emph{FH}), and the hashed forward algorithm with degree ordering  (\emph{FHD})—in comparison with the traditional \emph{MergePath} algorithm.

The observed performance improvement is remarkably significant. Specifically, \emph{F} achieves a $8.6 \times$ speedup, while \emph{FH} and \emph{FHD} achieve even more substantial speedups at $28.7 \times$ and $29.1 \times$, respectively. These results indicate that reducing the sizes of intersection sets and employing hash functions and degree ordering can collectively contribute to performance enhancements.

Similar to the hash method, degree ordering demonstrates substantial performance improvements across various scenarios. However, for roadNet-CA (Graph ID=20), roadNet-PA (Graph ID=21), and roadNet-TX (Graph ID=22), the hash-based algorithm \emph{FH} performs worse than the non-hashed algorithm \emph{F}, and the performance of degree ordering \emph{FD} is inferior to that of \emph{MergePath}. This arises from the fact that most vertices in these graphs possess similar and small numbers of degrees. Consequently, reordering the vertices has minimal impact on intersection performance and introduces additional overhead. Despite these exceptions, the combined approach of reducing intersection set sizes, hash functions, and degree ordering consistently enhances performance for a wide range of cases.

\begin{figure}
    \centering
    \includegraphics[width=0.895\textwidth]{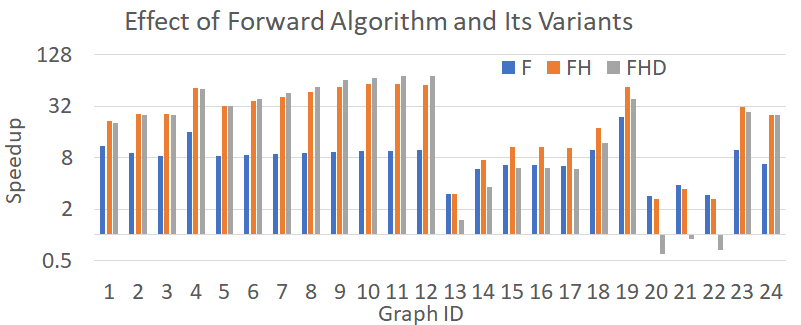}
    \caption{The speedups of Forward Algorithm and its variants compared with the \emph{MergePath} method.}
    \label{fig:SIS}
\end{figure}

\subsubsection{Effect of \emph{CETC-Seq} Algorithm and Its Variants} 
\label{subsec:NRExp}

The fundamental principle underlying the cover-edge method is minimizing unnecessary set intersection operations. In Fig. \ref{fig:CETC}, we illustrate the impact of the \emph{CETC-Seq} algorithm and its variants, namely \emph{CETC-Seq-D}, \emph{CETC-Seq-FE}, \emph{CETC-Seq-S}, \emph{CETC-Seq-SD}, \emph{CETC-Seq-SR}.

\begin{figure}
    \centering
    \includegraphics[width=0.895\textwidth]{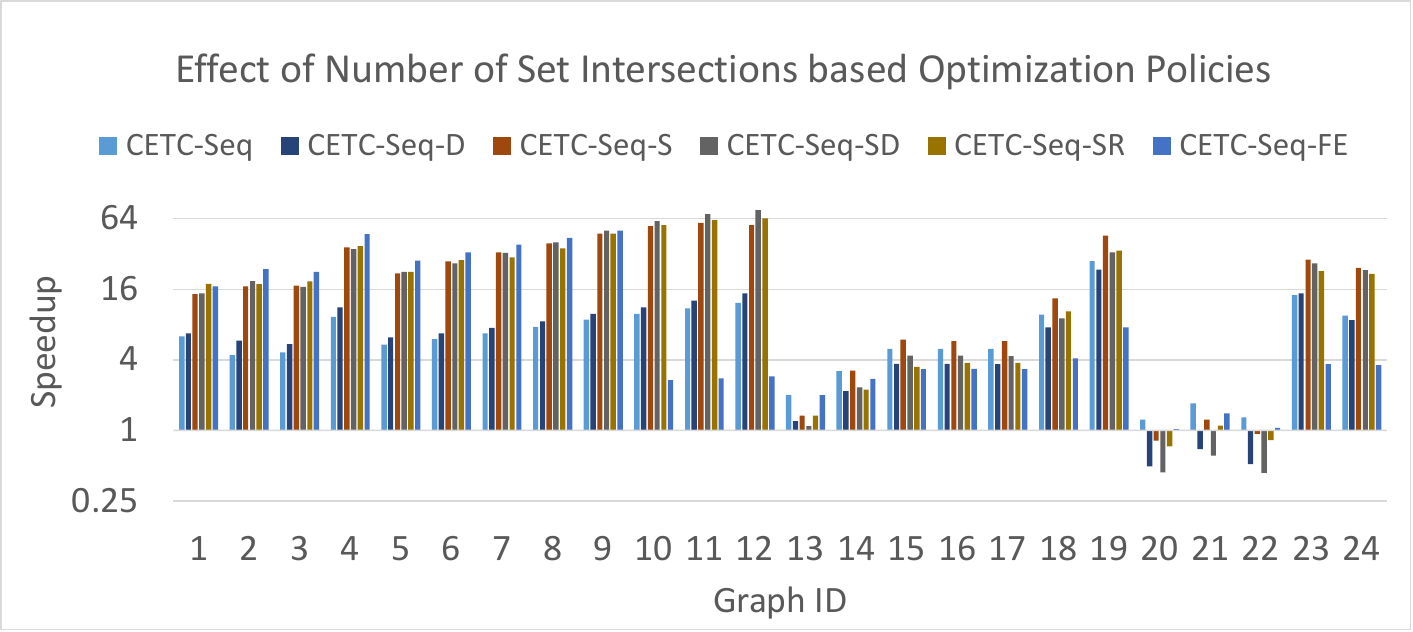}
    \caption{The speedups of \emph{CETC-Seq} and its variants compared with the \emph{MergePath} method.}
    \label{fig:CETC}
\end{figure}

Compared to the \emph{MergePath} method, \emph{CETC-Seq} demonstrates an average speedup of $7.4 \times$. \emph{CETC-Seq-D} achieves a slightly lower average speedup of $7.39 \times$, mainly due to its low performance on the road networks. 

\emph{CETC-Seq-FE} combines \emph{CETC-Seq} and \emph{F} in a unique manner. It employs \emph{CETC-Seq} for large graphs or when the $c$ value is small; otherwise, it uses \emph{F}. This switching approach yields an average speedup of $14.6 \times$. The rationale behind \emph{CETC-Seq-FE} lies in dynamically selecting the most suitable algorithm based on its compatibility with the characteristics of the graphs.

\emph{CETC-Seq-S} splits a graph into two parts based on the vertex levels marked by a \emph{BFS} pre-processing and applies \emph{CETC-Seq} and \emph{F} on each part. The performance of \emph{CETC-Seq-S} achieves a speedup of $23.4 \times$. This represents a more efficient combination method. Additionally, when we integrate degree ordering into \emph{CETC-Seq-S}, the resulting \emph{CETC-Seq-SD} algorithm performs slightly better than \emph{CETC-Seq-S}, achieving a speedup of $24.0 \times$. This result highlights that degree ordering works well with the \emph{F} algorithm. The reason is that degree ordering can further reduce the size of intersecting sets of the \emph{F} algorithm. \emph{CETC-Seq-SR} employs the recursive method to simplify the problem. For a large graph, it recursively applies \emph{CETC-Seq} to minimize set intersections, counting only triangles including non-horizontal edges, and finally applies \emph{F} to the smaller graph consisting of all horizontal edges that are known to include all the other triangles. \emph{CETC-Seq-SR} achieves an average speedup of $22.7 \times$. 

Notably, \emph{CETC-Seq} exhibits low performance on certain graphs compared to other methods. The relatively high overhead of \emph{BFS} preprocessing in \emph{CETC-Seq}, compared with set intersection, contributes to the low efficiency. A breakdown time analysis reveals that the percentages of \emph{BFS} processing time are 60\% of the total execution time. In the case of a long-diameter graph where each vertex has a small number of neighbors, the overhead of \emph{BFS} becomes large despite its time complexity of $\mathcal{O}(m)$ compared to the time complexity of total set intersections at $\mathcal{O}(m^{1.5})$. This overhead becomes particularly impactful when the neighbors of each vertex are limited, and the graph diameter is large. For road networks characterized by very small vertex degrees, where degree ordering introduces additional overhead without providing any significant benefit, \emph{CETC-Seq-D} experiences further performance degradation.

\subsubsection{Comprehensive Sequential Algorithms Comparison} \label{subsec:sequential}

In Fig. \ref{fig:sequential}, we present the relative execution time for twenty-two sequential triangle counting algorithms. If some triangle searching operations cannot find any triangle, we name them as fruitless operations here. 

While optimal in time complexity, the \emph{IR} spanning tree-based triangle counting algorithm exhibits nearly the slowest performance among all the compared algorithms. This is due to the involvement of spanning tree generation, removal of tree edges, and regeneration of a smaller graph in each iteration. Although these operations can be completed in $\mathcal{O}(m)$ time, the cost is relatively high in terms of practical performance.

The \emph{W} wedge-checking-based triangle counting algorithm often performs poorly. This is primarily because most graphs are sparse, resulting in that most wedge-checking operations are fruitless, or most wedges cannot form a triangle. For example, for the  RMAT 6 graph, the percentage of wedges/triangles is 0.53\%. For the RMAT 14 graph, the percentage reduces to 0.009\%. This makes most of the checks useless for counting triangles. \emph{W} can demonstrate better performance only when most of the graph's wedges can form triangles. This scenario is not common in most practical applications.

The algorithmic structures of \emph{EM} (Edge Merge Path), \emph{EB} (Edge Binary Search), \emph{ET} (Edge Partitioning), and \emph{EH} (Edge Hash) are very similar to each other, differing primarily in the set intersection methods they employ. Merge path requires pre-sorted adjacency lists, enabling it to compare the two adjacency lists of a given edge $\edge{u,v}$ in $d(u) + d(v)$ time. This is optimal because we have to check every neighbor. Binary search method \emph{EB} searches each vertex in a small adjacency list (e.g., $N(u)$) in a larger adjacency list (e.g., $N(v)$) in $d(u) \times \log(d(v))$ time.  \emph{ET} is a specific case of \emph{EB} and involves additional operations to find the midpoint of the two adjacency lists. Thus, from an algorithmic analysis perspective, \emph{ET}'s performance will always be worse than \emph{EB}'s. However, \emph{EB} and \emph{ET} can leverage parallelism effectively to improve performance. Our parallel results demonstrate that they may outperform \emph{EM}. \emph{EH} takes $min(d(u),d(v)) < d(u) + d(v)$ operations to find triangles, and the \emph{Hash} method doesn't require pre-sorting adjacency lists, making it better than \emph{EM} and often the best performer among the four methods.

\begin{figure*}
    \centering
    \includegraphics[width=1.0\textwidth]{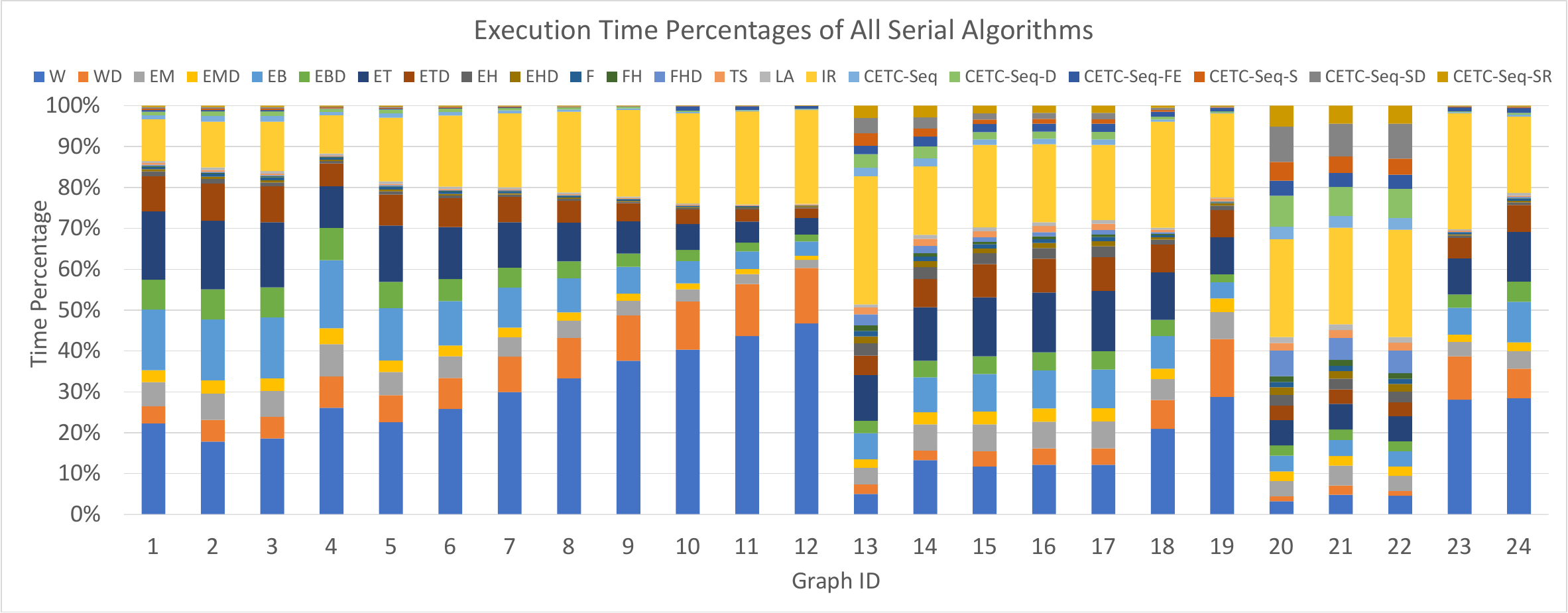}
    \caption{Percentage of execution time for different sequential triangle counting algorithms.}
    \label{fig:sequential}
\end{figure*}

\emph{TS} (Triangle Summation) and \emph{LA} (Linear Algebra) are two linear algebra-based methods. They can count the total number of triangles but cannot list all the triangles. Their performance improvements depend on optimizing formulas and architecture-related methods. The advantage of such methods lies in their ability to directly apply results from linear algebra theory and leverage highly optimized numerical techniques integrated into linear algebra libraries. Their performance is often superior to that of the \emph{EM} method.

\emph{F} (Forward) often demonstrates excellent performance in most scenarios but is inherently sequential. As we discussed earlier, \emph{F} dynamically generates two sets that are much smaller than the size of the original adjacency lists. It is based on the \emph{DO} method, which further reduces the fruitless checks in triangle counting operations. Additionally, pre-sorting vertices in non-increasing degrees enhance memory access locality and cache hit ratios. As one can observe, \emph{F} effectively reduces the operations that cannot find new triangles.  The results of \emph{FH} and \emph{FHD} show that the performance is further improved when \emph{Hash} is used.

\emph{CETC-Seq} and its variants, introduce another perspective for eliminating fruitless searches in triangle counting. First, it skips unnecessary edge searches based on a quick \emph{BFS} operation that can be completed in $\mathcal{O}(m + n)$ time. 
By leveraging the directed-oriented technique, \emph{CETC-Seq} achieves a further significant reduction in the fruitless searches during triangle counting. It is competitive with the fastest approaches and may be useful when the \emph{BFS} preprocessing overhead can be negligible. \emph{CETC-Seq-S} and its variants further optimize the performance with \emph{Hash}, degree ordering and recursive method. 

We assigned rank values to each test case and calculated the average rank value. The performance from high to low are  \emph{FH}, \emph{CETC-Seq-S}, \emph{FHD}, \emph{CETC-Seq}, \emph{LA}, \emph{F}, \emph{EHD}, \emph{TS}, \emph{CETC-Seq-SR}, \emph{CETC-Seq-SD}, \emph{CETC-Seq-FE}, \emph{EH}, \emph{CETC-Seq-D}, \emph{EMD}, \emph{EBD}, \emph{ET}, \emph{WD}, \emph{EB}, \emph{EM}, \emph{ETD}, \emph{W}, \emph{IR}.

We can say the top 4 set intersection-based triangle counting algorithms include our novel \emph{CETC-Seq-S} and \emph{CETC-Seq} algorithms.  The performance of \emph{CETC-Seq-S} with an average rank of 2.80 is slightly worse than that of \emph{FH} with an average rank of 2.0. The average rank of \emph{FHD} is 4.6 and \emph{CETC-Seq} is 6.4.

\subsubsection{Influence of the $c$ Value on the Performance of the Novel Algorithm}

Building upon the definition of our novel algorithm, its performance should be highly related to the covering ratio $c$. 

A noteworthy trend is identified when evaluating the results, particularly concerning the RMAT graphs. Our finding reveals that the forward algorithm and its variants tend to perform the fastest. As the scale of the RMAT graph increases, the parameter $c$ decreases, indicating a more substantial removal of fruitless checks after \emph{BFS}. Under these conditions, our novel method demonstrates greater efficiency compared to the \emph{F} algorithms.

These observations validate our hypothesis that the performance of our new algorithm is significantly correlated with the covering ratio $c$. As $c$ decreases, performance improves. 

Concurrently, an analysis of the performance of the road network graphs (roadNet-CA, roadNet-PA, roadNet-TX) reveals their divergence from the other graphs. Road networks, unlike social networks, often have only low-degree vertices (for instance, many degree four vertices), and large diameters. Although the covering ratio of these road networks is under 15\%, we see less benefit from the new approach due to this low value of $c$. So, a lower $c$ value does not always yield high performance.

\subsection{Results and Analysis of Parallel Algorithms on Shared-Memory} \label{subsec:POExp}
The execution times of the parallel algorithms (in seconds) are presented in Table ~\ref{t:results:parallel32} for 32 threads and in Table ~\ref{t:results:parallel224} for 224 threads.

\newsavebox{\boxParallelthirtytwo}
\begin{lrbox}{\boxParallelthirtytwo}
\begin{tabular}{lrrrrrrrrrrr}
Graph          & WP         & WDP        & EMP       & EMDP      & EBP       & EBDP      & ETP       & ETDP      & EHP       & EHDP      & CETC-SM  \\ \hline
RMAT 6 & 0.000449	& 0.000268	& 0.000226	& 0.000137	& 0.000251	& 0.000149	& 0.000236	& 0.000163	& 0.00018	& 0.000102	& 0.000143 \\
RMAT 7 & 0.001223	& 0.000446	& 0.000218	& 0.000154	& 0.000298	& 0.000159	& 0.000284	& 0.000232	& 0.000153	& 0.000088	& 0.000112 \\
RMAT 8 & 0.004417	& 0.001139	& 0.000413	& 0.000259	& 0.000607	& 0.000314	& 0.000627	& 0.000486	& 0.000368	& 0.000207	& 0.00033 \\
RMAT 9 & 0.006898	& 0.002485	& 0.000661	& 0.000604	& 0.001263	& 0.000684	& 0.001419	& 0.001355	& 0.000492	& 0.000273	& 0.00043 \\
RMAT 10 & 0.018563	& 0.006674	& 0.001527	& 0.001515	& 0.003311	& 0.001667	& 0.003505	& 0.003399	& 0.000902	& 0.000511	& 0.000935 \\
RMAT 11 & 0.043502	& 0.020521	& 0.003952	& 0.003951	& 0.007708	& 0.0041	& 0.008889	& 0.008896	& 0.002429	& 0.001303	& 0.002447 \\
RMAT 12 & 0.117014	& 0.049646	& 0.005842	& 0.005817	& 0.011434	& 0.005798	& 0.012514	& 0.011956	& 0.00309	& 0.001837	& 0.003849 \\
RMAT 13 & 0.279936	& 0.107251	& 0.014796	& 0.013801	& 0.027184	& 0.013905	& 0.03255	& 0.028344	& 0.006466	& 0.003631	& 0.00932 \\
RMAT 14 & 0.738807	& 0.294576	& 0.035303	& 0.034693	& 0.064022	& 0.031186	& 0.078634	& 0.068282	& 0.015326	& 0.008446	& 0.024843 \\
RMAT 15 & 1.945767	& 0.883749	& 0.093881	& 0.090363	& 0.159664	& 0.080606	& 0.193024	& 0.167142	& 0.034382	& 0.017455	& 0.054523 \\ 
RMAT 16 & 5.445053	& 2.616904	& 0.232325	& 0.225416	& 0.388562	& 0.190387	& 0.471949	& 0.405698	& 0.075692	& 0.038662	& 0.123288 \\
RMAT 17 & 16.105178	& 7.952496	& 0.600006	& 0.570236	& 1.017046	& 0.495515	& 1.201604	& 1.01826	& 0.173762	& 0.085809	& 0.272745 \\
karate & 0.000047	& 0.000046	& 0.000045	& 0.000052	& 0.000046	& 0.000049	& 0.000047	& 0.000047	& 0.000046	& 0.000052	& 0.00005 \\
amazon0302 & 0.008796	& 0.022824	& 0.014149	& 0.036935	& 0.020339	& 0.039651	& 0.018675	& 0.019454	& 0.011521	& 0.029648	& 0.027349 \\
amazon0312 & 0.040122	& 0.072721	& 0.039475	& 0.115382	& 0.058515	& 0.125074	& 0.069724	& 0.086947	& 0.045718	& 0.099768	& 0.095591 \\
amazon0505 & 0.047963	& 0.084478	& 0.045625	& 0.14529	& 0.0728	& 0.147519	& 0.082892	& 0.090064	& 0.046591	& 0.10615	& 0.09897 \\
amazon0601 & 0.048284	& 0.083188	& 0.05004	& 0.134631	& 0.072389	& 0.149232	& 0.082411	& 0.091988	& 0.049378	& 0.107361	& 0.101416 \\
loc-Brightkit & 0.026786	& 0.018883	& 0.01209	& 0.023231	& 0.006994	& 0.012797	& 0.008975	& 0.008192	& 0.004159	& 0.005755	& 0.005673 \\
loc-Gowalla & 1.720423	& 0.544702	& 0.509203	& 0.076514	& 0.048047	& 0.986971	& 0.991823	& 0.111531	& 0.072055	& 0.027526	& 0.027009 \\ 
roadNet-CA & 0.006707	& 0.099629	& 0.041482	& 0.090839	& 0.054792	& 0.11503	& 0.063274	& 0.095634	& 0.063806	& 0.06235	& 0.059879 \\
roadNet-PA & 0.002327	& 0.045298	& 0.04421	& 0.051884	& 0.027293	& 0.050139	& 0.031387	& 0.046944	& 0.034754	& 0.033182	& 0.02988 \\
roadNet-TX & 0.00303	& 0.061665	& 0.043471	& 0.07115	& 0.040102	& 0.080569	& 0.04338	& 0.060169	& 0.041411	& 0.042741	& 0.040724 \\
soc-Epinions1 & 0.188828	& 0.029691	& 0.024934	& 0.044974	& 0.02401	& 0.065555	& 0.052619	& 0.02325	& 0.011699	& 0.018498	& 0.019213 \\
wiki-Vote & 0.020278	& 0.00745	& 0.004677	& 0.013179	& 0.007166	& 0.018188	& 0.011463	& 0.006211	& 0.003359	& 0.007699	& 0.007689
\end{tabular}
\end{lrbox}

\begin{table*}
\scriptsize
\centering
\caption{Execution time (in seconds) for shared-memory parallel algorithms. (32 threads)}
\scalebox{0.69}{\usebox{\boxParallelthirtytwo}}
\label{t:results:parallel32}
\end{table*}

\newsavebox{\boxParalleltwotwentyfour}
\begin{lrbox}{\boxParalleltwotwentyfour}
\begin{tabular}{lrrrrrrrrrrr}
Graph          & WP         & WDP        & EMP       & EMDP      & EBP       & EBDP      & ETP       & ETDP      & EHP       & EHDP      & CETC-SM  \\ \hline
RMAT 6 & 0.00262	& 0.001665	& 0.00159	& 0.002007	& 0.001779	& 0.002791	& 0.001626	& 0.002978	& 0.002411	& 0.001657	& 0.001823 \\
RMAT 7 & 0.005946	& 0.0012	& 0.001177	& 0.000863	& 0.001092	& 0.001765	& 0.001126	& 0.000937	& 0.000799	& 0.000756	& 0.000913 \\
RMAT 8 & 0.01391	& 0.004556	& 0.002609	& 0.002877	& 0.002534	& 0.002047	& 0.003943	& 0.002736	& 0.002399	& 0.002099	& 0.002797 \\
RMAT 9 & 0.032362	& 0.008507	& 0.003102	& 0.002114	& 0.002848	& 0.004165	& 0.005741	& 0.00493	& 0.002538	& 0.002196	& 0.002278 \\
RMAT 10 & 0.096783	& 0.023608	& 0.006077	& 0.004446	& 0.006682	& 0.007183	& 0.009634	& 0.009532	& 0.002827	& 0.002962	& 0.00238 \\
RMAT 11 & 0.244285	& 0.034017	& 0.005143	& 0.005314	& 0.008019	& 0.006563	& 0.012168	& 0.011372	& 0.002077	& 0.001958	& 0.002647 \\
RMAT 12 & 0.603094	& 0.054447	& 0.007437	& 0.007343	& 0.005859	& 0.00551	& 0.019386	& 0.019021	& 0.002241	& 0.001818	& 0.003157 \\
RMAT 13 & 1.767242	& 0.191662	& 0.017554	& 0.017532	& 0.013061	& 0.012225	& 0.048984	& 0.048837	& 0.004743	& 0.003995	& 0.005822 \\
RMAT 14 & 5.805617	& 0.6103	& 0.043436	& 0.043317	& 0.032064	& 0.025181	& 0.119665	& 0.124235	& 0.008927	& 0.00792	& 0.011054 \\
RMAT 15 & 19.502675	& 1.075846	& 0.1124	& 0.112662	& 0.06781	& 0.050401	& 0.284255	& 0.274718	& 0.021382	& 0.01982	& 0.026385 \\
RMAT 16 & 4.564897	& 2.724765	& 0.27826	& 0.273954	& 0.125003	& 0.10814	& 0.566664	& 0.548104	& 0.052475	& 0.046825	& 0.046742 \\
RMAT 17 & 14.404249	& 8.292626	& 0.636928	& 0.623498	& 0.359324	& 0.209283	& 1.255369	& 1.193338	& 0.126567	& 0.11914	& 0.10871 \\
karate & 0.000246	& 0.001053	& 0.001054	& 0.00139	& 0.001052	& 0.001067	& 0.001843	& 0.002898	& 0.002982	& 0.001061	& 0.001648 \\
amazon0302 & 0.028943	& 0.007874	& 0.016644	& 0.00981	& 0.011362	& 0.006467	& 0.013704	& 0.007236	& 0.035765	& 0.024087	& 0.039814 \\
amazon0312 & 0.10007	& 0.046193	& 0.029063	& 0.020546	& 0.030613	& 0.016588	& 0.067293	& 0.053589	& 0.065335	& 0.038084	& 0.084911 \\
amazon0505 & 0.134932	& 0.043264	& 0.03243	& 0.021458	& 0.032284	& 0.017594	& 0.070966	& 0.055361	& 0.071815	& 0.039119	& 0.090892 \\
amazon0601 & 0.111001	& 0.050005	& 0.031588	& 0.02097	& 0.032657	& 0.017569	& 0.071926	& 0.053841	& 0.063443	& 0.035649	& 0.092445 \\
loc-Brightkit & 0.298194	& 0.037024	& 0.007619	& 0.005149	& 0.005186	& 0.002977	& 0.009518	& 0.009214	& 0.008017	& 0.003441	& 0.005172 \\
loc-Gowalla & 4.642555	& 1.850222	& 0.551523	& 0.597855	& 0.042251	& 0.035543	& 0.99552	& 1.153058	& 0.142559	& 0.135726	& 0.02852 \\
roadNet-CA & 0.018777	& 0.008275	& 0.034299	& 0.018575	& 0.029132	& 0.016644	& 0.029072	& 0.016026	& 0.033069	& 0.021434	& 0.094693 \\
roadNet-PA & 0.012264	& 0.005497	& 0.02696	& 0.015148	& 0.016778	& 0.013736	& 0.016576	& 0.010023	& 0.020221	& 0.012846	& 0.048229 \\
roadNet-TX & 0.013831	& 0.005999	& 0.026138	& 0.01613	& 0.024868	& 0.011356	& 0.019374	& 0.011121	& 0.024195	& 0.014831	& 0.058367 \\
soc-Epinions1 & 2.622935	& 0.322065	& 0.032298	& 0.029065	& 0.022567	& 0.01809	& 0.087574	& 0.08822	& 0.014144	& 0.00968	& 0.010896 \\
wiki-Vote & 0.527084	& 0.028837	& 0.005353	& 0.00425	& 0.007596	& 0.005707	& 0.0129	& 0.008829	& 0.002952	& 0.001597	& 0.003614 
\end{tabular}
\end{lrbox}

\begin{table*}
\scriptsize
\centering
\caption{Execution time (in seconds) for shared-memory parallel algorithms. (224 threads)}
\scalebox{0.69}{\usebox{\boxParalleltwotwentyfour}}
\label{t:results:parallel224}
\end{table*}

\subsubsection{Performance Utilizing 32 Threads}
While \emph{F} and its variants excel as sequential algorithms, they are inherently sequential and cannot be parallelized. In this section, we focus on algorithms conducive to parallelization to showcase the speedups achieved with parallel methods. Fig. \ref{fig:PP} illustrates the speedups of various parallel algorithms compared to their corresponding sequential counterparts, employing 32 threads.

The average speedups are as follows: \emph{WP} is $10.5 \times$; \emph{WDP} is $7.5 \times$; \emph{EMP} is $13.6 \times$; \emph{EMDP} is $8.3 \times$; \emph{EBP} is $23.3 \times$; \emph{EBDP} is $19.3 \times$; \emph{ETP} is $16.2 \times$; \emph{ETDP} is $10.8 \times$; \emph{EHP} is $6.6 \times$; \emph{EHDP} is $5.0 \times$; \emph{CETC-SM} is $3.9 \times$. The results affirm that parallel optimization significantly improves performance.

However, certain scenarios highlight limitations. For instance, in the case of the small-sized graph ``karate'' (Graph ID=13), all parallel algorithms fail to exhibit performance improvements. This can be attributed to the inherent overhead of the OpenMP parallel method, which outweighs the benefits for very small graphs. A similar pattern is observed for the graph RMAT 6 (Graph ID=1), where three parallel methods—\emph{EHP}, \emph{EHDP}, and \emph{CETC-SM}—show no performance improvement. As previously mentioned, the baseline algorithms \emph{EH}, \emph{EHD}, and \emph{CETC-Seq} have already demonstrated high performance, and the parallel overhead for small graphs nullifies the potential benefits of parallelization.

\begin{figure*}
    \centering
    \includegraphics[width=0.89\textwidth]{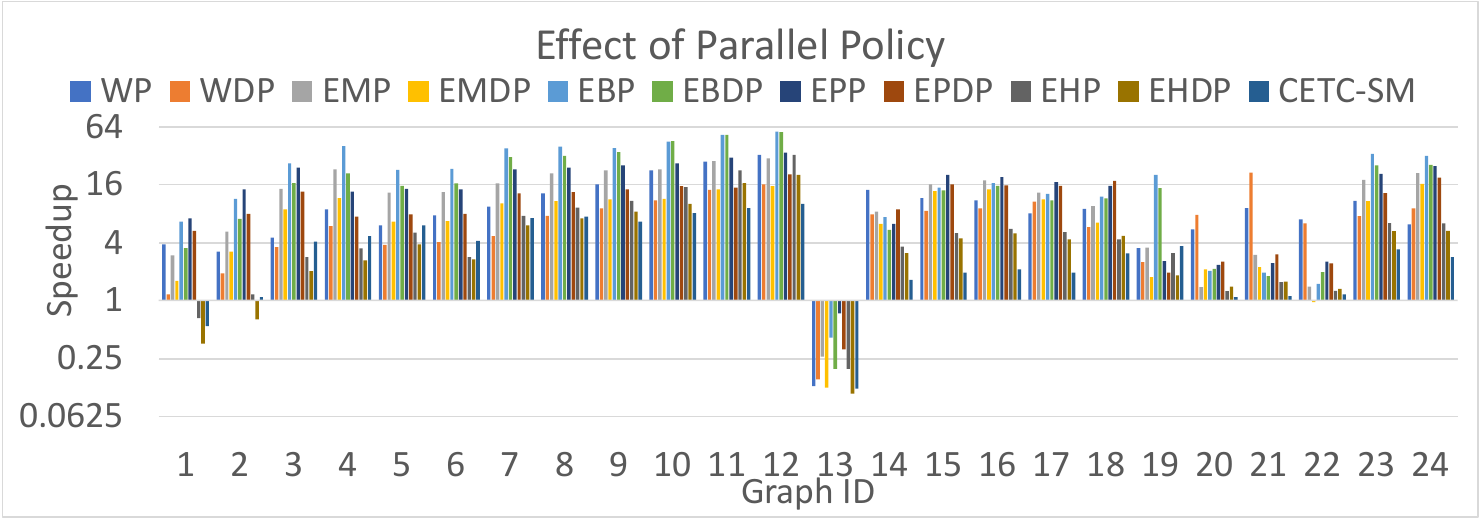}
    \caption{The speedups of parallel optimization methods compared with their sequential counterparts using 32 threads.}
    \label{fig:PP}
\end{figure*}

\subsubsection{Performance Utilizing All System Threads}
When we harness our experimental system's full parallel processing capacity, we can execute our OpenMP parallel programs with 224 threads. In Fig. \ref{fig:parallel224}, we provide the execution time percentages of various algorithms when employing all 224 system threads. We assigned rank values to each test case and calculated the average rank value. The performance from high to low are  \emph{EHDP}, \emph{EBDP}, \emph{EMDP}, \emph{EHP}, \emph{EBP}, \emph{CETC-SM}, \emph{EMP}, \emph{ETDP}, \emph{WDP}, \emph{ETP}, \emph{WP}.

The presented results highlight that not only can hash and binary search deliver commendable parallel performance by minimizing operations per parallel thread but also the application of degree ordering proves effective in improving the performance of individual threads.

\begin{figure*}
    \centering
    \includegraphics[width=0.895\textwidth]{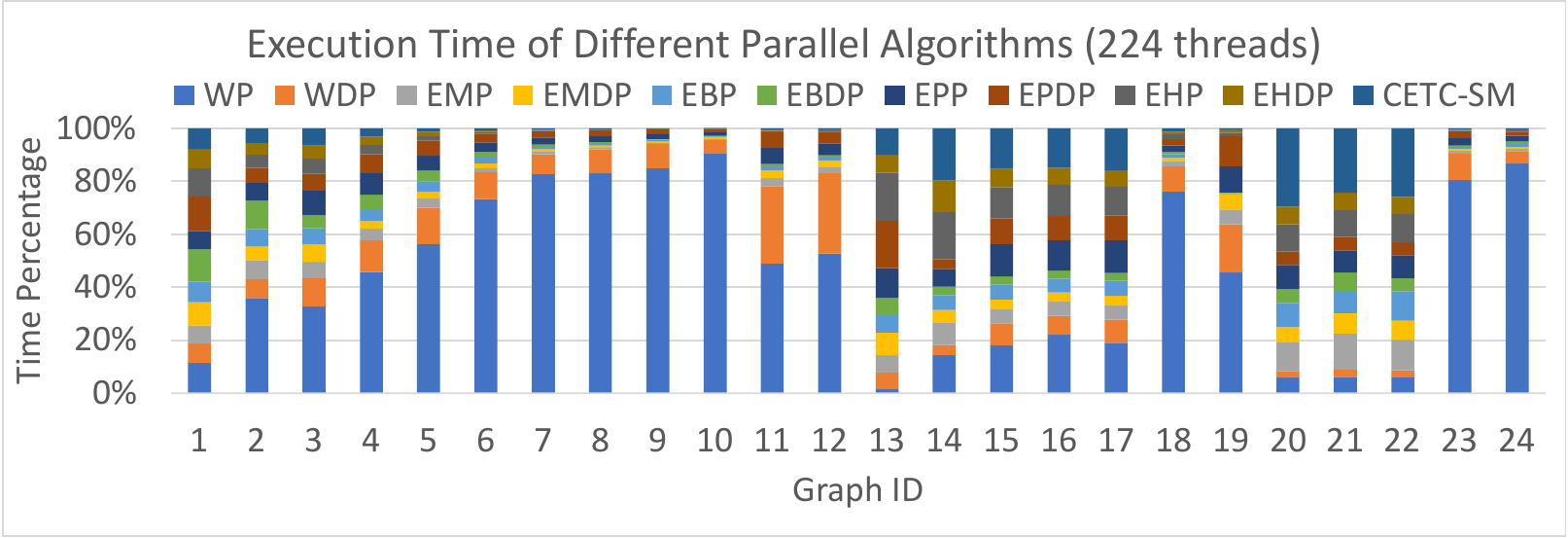}
    \caption{The execution time percentage of different parallel triangle counting algorithms with all 224 parallel threads/cores.}
    \label{fig:parallel224}
\end{figure*}

\subsubsection{Scalable Performance}
This subsection delves into the performance of these algorithms in response to varying thread counts. We use RMAT 15 as an illustrative example of a synthetic graph and Amazon0312 as a representative instance of a real graph. By progressively increasing the number of threads to 2, 4, 8, 16, 32, 64, 128, and 224, we seek to identify changes in speedup corresponding to increasing thread counts.

Fig. \ref{fig:rmat15} illustrates the change in speedup with the increasing number of threads on RMAT15. For most algorithms, a bottleneck emerges starting from 64 threads, with no discernible speedup observed with the continued increase in thread count. Notably, the \emph{WP} algorithm exhibits a degradation in performance with the incorporation of additional parallel threads. The only algorithm demonstrating notable scalability is \emph{EBP}, showcasing consistent performance improvement with the increasing number of threads. A similar observation is made for the real graph (see Fig. \ref{fig:amazon0312}), where most algorithms encounter a bottleneck at 64 threads. However, \emph{EBP} and \emph{EBDP} exhibit good scalability, indicating that binary search-based methods possess superior scalability compared to other approaches.

\begin{figure}
\begin{minipage}{.49\textwidth}
    \centering
    \includegraphics[width=.99\linewidth]{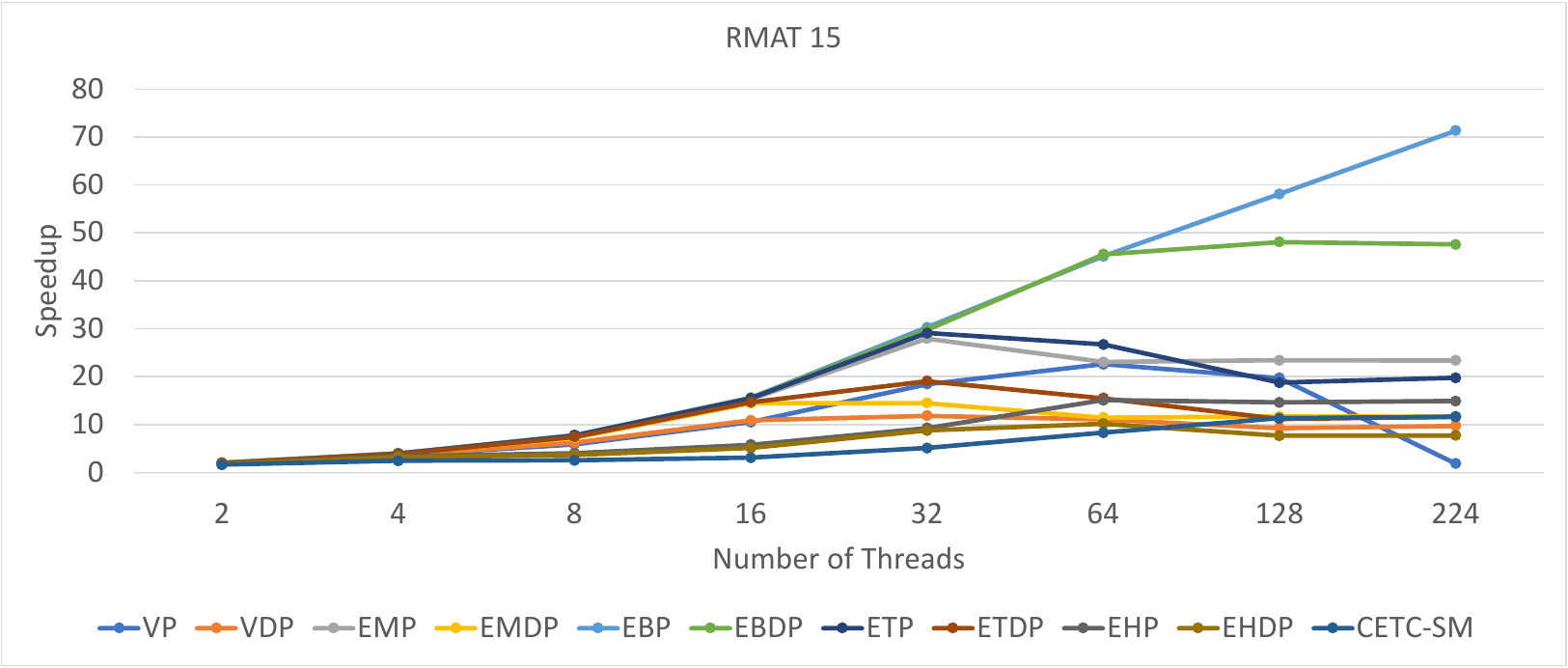}
    \caption{Speedups of various algorithms on RMAT15 compared with a single-thread setup}
    \label{fig:rmat15}
\end{minipage}
\begin{minipage}{.49\textwidth}
    \centering
    \includegraphics[width=.86\linewidth]{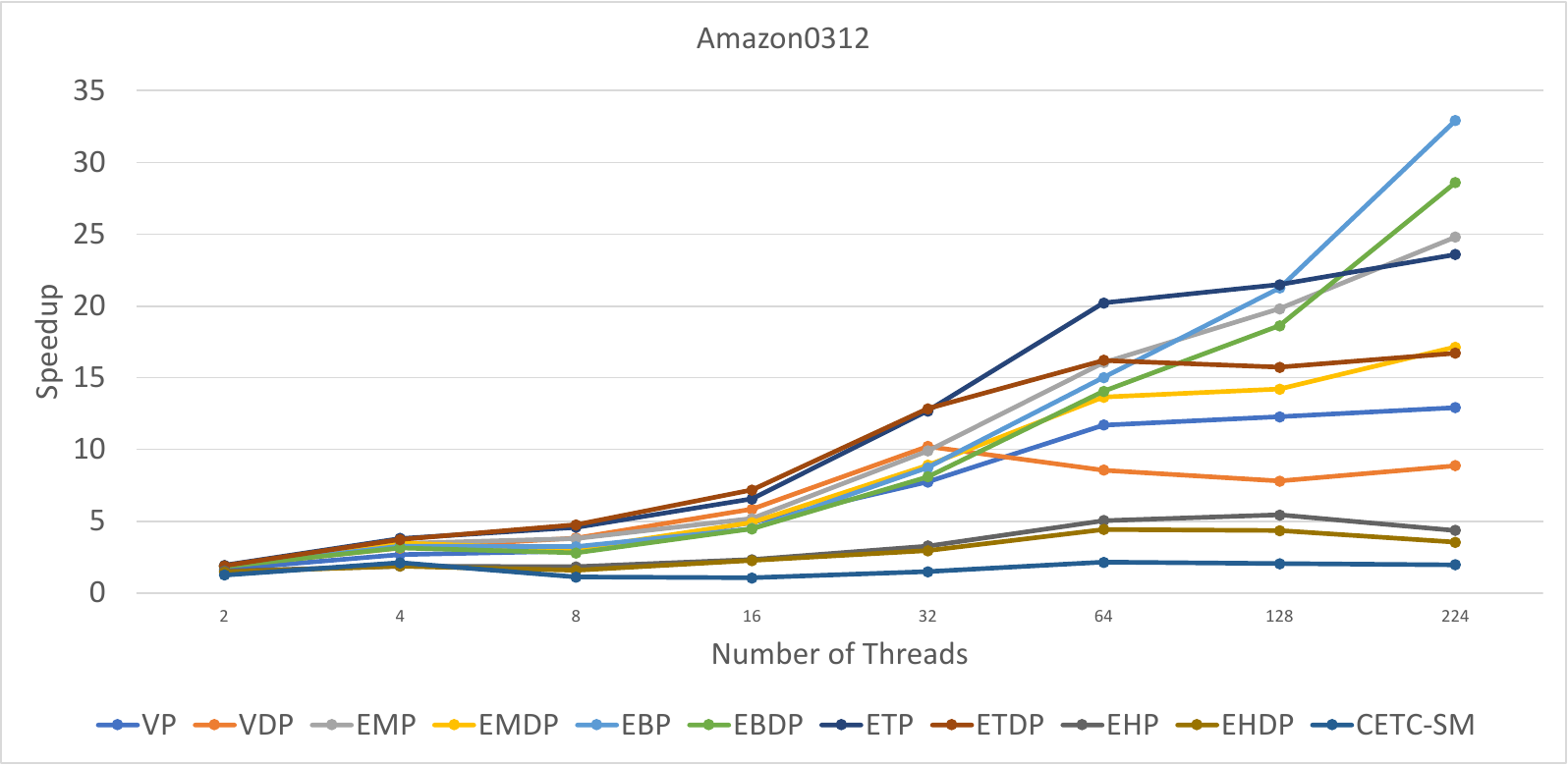}
    \caption{Speedups of various algorithms on Amazon0312 compared with a single-thread setup}
    \label{fig:amazon0312}
    \end{minipage}
\end{figure}

\subsubsection{Best Performance on Different Graphs} \label{subsec:best}
In this section, we use \emph{EM} as the performance baseline to evaluate the best speedup achieved by different algorithms. The results are summarized in Table \ref{tab:best}. The number following a specific algorithm name indicates how many parallel threads are used.

Integrated optimization methods demonstrate a substantial speedup, averaging at 75.8. Examining various algorithms on different graphs unveils insights into optimization methods.

Firstly, for small graphs like RMAT 6 (Graph ID=1), RMAT 7 (Graph ID=2), and karate (Graph ID=13), parallel optimization techniques fail to outperform the sequential \emph{FH} and linear algebra \emph{LA} methods. Practical performance considerations suggest that employing multiple parallel threads might introduce overhead for small graphs, making sequential methods more efficient.

Secondly, as graph size increases, optimal performance often requires more parallel resources. However, beyond a critical point, additional parallel resources may lead to decreased performance. For example, RMAT 8 (Graph ID=3) and roadNet-TX (Graph ID=22) achieve peak performance with 32 threads. In contrast, RMAT 9 to RMAT 10 (Graph ID 4-5), RMAT 12 (Graph ID=7), RMAT 14 to RMAT 17 (Graph ID 9-12), loc-Brightkite  (Graph ID=18), and roadNet-PA (Graph ID=21) require 64 threads. Certain graphs, such as RMAT 11 (Graph ID=6), RMAT 13 (Graph ID=8), loc-Gowalla (Graph ID=19), roadNet-CA (Graph ID=20), soc-Epinions1 (Graph ID=23), wiki-Vote (Graph ID=24), demand 128 threads. Larger graphs like amazon0302, amazon0312, amazon0505, and amazon0601 (Graph ID 14-17) leverage the full system parallel resources (224 threads). Notably, graph size alone doesn't determine parallel resource needs, as topology plays a crucial role in parallel performance. At the same time, the parallel algorithms that achieve the best performance are also different.  Among all the parallel algorithms, \emph{EHDP} has 13 times to achieve the best performance. \emph{EBDP} has four times to achieve the best performance. \emph{WDP} has three times to achieve the best performance and \emph{CETC-SM} has one time to achieve the best performance.

Thirdly, the various sequential and parallel optimizations needed for better performance can differ. For instance, \emph{WD} might not be ideal in a sequential scenario due to checking numerous wedges, many of which are not fruitful for sparse graphs. However, in a parallel scenario, \emph{WDP64} excels with 64 threads on roadNet-PA (Graph ID=21), surpassing other algorithms. The efficiency arises from the smaller number of wedges when vertex degrees are low, coupled with \emph{DO} optimization method that reduces fruitless searches. Another case is \emph{EBD}, which may not be favorable in sequential algorithms due to increased total operations compared to \emph{EMD}. However, in parallel algorithms, \emph{EBDP} could outperform \emph{MergePath} by distributing work more efficiently through parallel binary searches.

In conclusion, our results highlight that different algorithms find their optimal scenarios based on specific graph topology and hardware configurations. Graph topology and available hardware resources are pivotal factors in selecting the most efficient triangle counting algorithm.

\newsavebox{\boxBest}
\begin{lrbox}{\boxBest}
\begin{tabular}{|c|c|c|c|c|c|c|c|c|c|c|c|c|}
\hline
Graph     & 1         & 2         & 3         & 4         & 5         & 6         & 7         & 8         & 9         & 10        & 11        & 12         \\ \hline
Baseline Time & 0.0006080 & 0.0019730 & 0.0054550 & 0.0146430 & 0.0203420 & 0.0537350 & 0.1411760 & 0.3726090 & 0.9877480 & 2.6268370 & 6.9313980 & 18.3285120 \\ \hline
Best Time     & 0.0000280 & 0.0000760 & 0.0002070 & 0.0002170 & 0.0003830 & 0.0008570 & 0.0014210 & 0.0029630 & 0.0066920 & 0.0150250 & 0.030892  & 0.078430   \\ \hline
Algorithm & FH        & FH        & EHDP32    & EHDP64    & EHDP64    & EHDP128   & EHDP64    & EHDP128   & EHDP64    & EHDP64    & EHDP64    & EHDP64     \\ \hline
Speedup   & 21.7      & 26.0      & 26.4      & 67.5      & 53.1      & 62.7      & 99.3      & 125.8     & 147.6     & 174.8     & 224.4     & 233.7      \\ \hline
\hline
Graph     & 13        & 14        & 15        & 16        & 17        & 18        & 19        & 20        & 21        & 22        & 23        & 24         \\ \hline
Baseline Time & 0.0000120 & 0.1434740 & 0.7208080 & 0.7557200 & 0.7621240 & 0.1158030 & 2.0830490 & 0.1027660 & 0.0755300 & 0.0708210 & 0.6153270 & 0.1243600  \\ \hline
Best Time     & 0.0000020 & 0.0064670 & 0.0165880 & 0.0175940 & 0.0175690 & 0.0028280 & 0.0200290 & 0.002870  & 0.001671  & 0.003030  & 0.0088180 & 0.0015570  \\ \hline
Algorithm & LA        & EBDP224   & EBDP224   & EBDP224   & EBDP224   & EHDP64    & CETC-SM128    & WDP128    & WDP64     & WDP32     & EHDP128   & EHDP128    \\ \hline
Speedup   & 6.0       & 22.2      & 43.5      & 43.0      & 43.4      & 40.9      & 104.0     & 35.8      & 45.2      & 23.4      & 69.8      & 79.9       \\ \hline
\end{tabular}
\end{lrbox}

\begin{table*}
\scriptsize
\centering
\caption{Best Performance and Algorithms for Different Graphs (second).}
\scalebox{0.57}{\usebox{\boxBest}}
\label{tab:best}
\end{table*}

\subsection{Communication Analysis of CETC-DM}

\newsavebox{\boxResult}
\begin{lrbox}{\boxResult}
\begin{tabular}{|l|r|r|r|r|r|r|r|r|r|}
    \hline
        \textbf{Graph} & \textbf{n} & \textbf{m} & \textbf{\# Triangles} & \textbf{\# Wedges} & \textbf{$c$} & \textbf{$p$} & \textbf{Previous} & \textbf{This paper} & \textbf{Reduction} \\ \hline
        ca-GrQc & 5242 & 14484 & 48260 & 165798 & 0.522 & 4 & 526KB & 122KB & 4.31 \\ \hline
        ca-HepTh & 9877 & 25973 & 28339 & 277389 & 0.423 & 4 & 948KB & 218KB & 4.35 \\ \hline
        as-caida20071105 & 26475 & 53381 & 36365 & 776895 & 0.225 & 4 & 2.78MB & 401KB & 7.10 \\ \hline
        facebook\_combined & 4039 & 88234 & 1612010 & 17051688 & 0.914 & 4 & 48.8MB & 893KB & 56.0 \\ \hline
        ca-CondMat & 23133 & 93439 & 173361 & 1567373 & 0.511 & 4 & 5.61MB & 897KB & 6.40 \\ \hline
        ca-HepPh & 12008 & 118489 & 3358499 & 5081984 & 0.621 & 4 & 17.0MB & 1.13MB & 15.1 \\ \hline
        email-Enron & 36692 & 183831 & 727044 & 5933045 & 0.478 & 4 & 22.6MB & 1.79MB & 12.7 \\ \hline
        ca-AstroPh & 18772 & 198050 & 1351441 & 8451765 & 0.667 & 4 & 30.2MB & 2.08MB & 14.6 \\ \hline
        loc-brightkite\_edges & 58228 & 214078 & 494728 & 6956250 & 0.441 & 4 & 26.5MB & 2.02MB & 20.4 \\ \hline
        soc-Epinions1 & 75879 & 405740 & 1624481 & 21377935 & 0.498 & 4 & 86.7MB & 4.25MB & 10.7 \\ \hline
        amazon0601 & 403394 & 2443408 & 3986507 & 96348699 & 0.529 & 8 & 436MB & 40.9MB & 10.7 \\ \hline
        
        com-Youtube & 1134890 & 2987624 & 3056386 & 209811585 & 0.347 & 8 & 1.03GB & 44.3MB & 23.7 \\ \hline
        RMAT-36 & 68719476736 & 1099511627776 & \emph{1.2E+14} & \emph{2.73E+16} & \emph{0.311} & 128 & 218PB & \emph{192TB} & \emph{1156} \\ \hline
        RMAT-42 & 4398046511104 & 70368744177664 & \emph{1.3E+16} & \emph{5.79E+18} & \emph{0.260} & 256 & 52.8EB & \emph{22.8PB} & \emph{2368} \\ \hline
\end{tabular}
\end{lrbox}

\begin{table*}
\scriptsize
\centering
\caption{Communication costs of \emph{CETC-DM} for real and synthetic graph. The synthetic graphs are Graph500 RMAT graphs of scale 36 and 42. The column \textbf{`Previous'} represents the communication volume of the best prior parallel algorithms \cite{dolev2012tri,pearce2018k,sanders2023engineering}, that use wedge-checking based algorithms and \textbf{`This paper'} represents the communication cost of our new approach \emph{CETC-DM}. \textbf{`Reduction'} represents the communication reduction between these two, and thus, the expected speedup of the parallel algorithm. Entries in \emph{italics} are estimated values.}
\scalebox{0.68}{\usebox{\boxResult}}
\label{tab:results}
\end{table*}

In this section, we analyze the performance of the parallel triangle counting algorithm \emph{CETC-DM} (Alg.~\ref{algparallel}) on both real and synthetic graphs. We implemented our new triangle counting algorithm using Python to accurately compute the exact communication volume and determine an analytic model based on the size of the graph and number of processors, and the covering ratio ($c$) from the BFS.  The results given in Table~\ref{tab:results} are exact communication volumes from our new algorithm on all of the graphs except the two large RMAT graphs where we compute the communication volume from the validated analytic model. For the comparison with prior approaches \cite{dolev2012tri,pearce2018k,sanders2023engineering}, we estimate the communication volume from the number of wedges which is exact for all graphs other than the last two large RMAT graphs where we estimate the number of wedges using graph theory.

For the real graphs, we find the actual value of $c$, the percentage of graph edges that are cover-edges, for an arbitrary breadth-first search, and set the number $p$ of processors to a reasonable number given the size of the graph. For the synthetic graphs, we use large Graph500 RMAT graphs \cite{chakrabarti2004r} with parameters $a = 0.57$, $b = 0.19$, $c = 0.19$, and $d = 0.05$, for scale 36 and 42 with  $n=2^{\mbox{scale}}$ and $m=16n$, similar with the IARPA AGILE benchmark graphs, and set $p$ according to estimates of potential system sizes with sufficient memory to hold these large instances. 

For comparison, most prior parallel algorithms for triangle counting operate on the graph as follows.
A parallel loop over the vertices $v \in V$ produces all 2-paths (\emph{wedges}) where $\edge{v,v_1} , \edge{v,v_2} \in E$ and (w.l.o.g.) $v_1<v_2$. The processor that produces this wedge will send an open wedge query message containing the vertex IDs of $v_1$ and $v_2$ to the processor that owns vertex $v_1$. If the consumer processor that receives this query message finds an edge $\edge{v_1,v_2} \in E$, then a local triangle counter is incremented. After producers and consumers complete all work, a global reduction over the $p$ triangle counts computes the total number of triangles in $G$.

\subsubsection{Graph500 RMAT Graphs}

\begin{figure}
    \centering
    \includegraphics[width=0.5\textwidth]{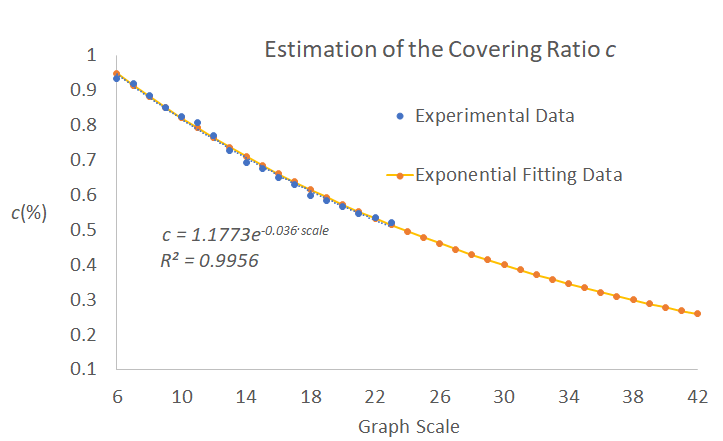}
    \caption{Estimate of $c$ using an exponential model, based on observations of $c$ for RMAT graph scale 6 to 23 graphs.}
    \label{fig:c}
\end{figure}

For the large  Graph500 RMAT graphs, the number of triangles is estimated from our model based on the number of triangles found in RMAT graphs up to scale 29 in the literature \cite{hoang2019disttc,giechaskiel2015pdtl,chakrabarti2004r,burkhardt2017graphing}. The fitting equation is $\mbox{\#Triangles} = 77.422n^{1.125}$ with $R^2 = 1.0$, where $n$ is the total number of vertices. The number of triangles estimated for scale 36 and 42 RMAT graphs are $1.20 \times 10^{14}$ and $1.30 \times 10^{16}$, respectively.

We estimate the number of wedges for the scale 36 and 42 Graph500 RMAT graphs based on the theorem given by Seshadhri \emph{et al.} in \cite{seshadhri2011hitchhiker}. According to their formula, we can estimate the expected number of vertices \(N(d)\) for a given out-degree \(d\). The number of wedges that can be formed by vertices with such a degree is calculated as \(\binom{d}{2} \times N(d)\), where \(\binom{d}{2}\) means choosing two from \(d\).

By summing all such wedges generated from the minimum ($ e\ln n $)  to the maximum degree (\(\sqrt{n}\)), which is the assumption of the formula,  we can approximate the total number of wedges in the given graph, where \(n\) is the total number of vertices. This is a conservative estimate because it only considers the out-degree instead of the sum of out and in-degrees. Employing the formula, we calculate the number of wedges to be \(2.73 \times 10^{16}\) for scale 36 and \(5.8 \times 10^{18}\) for scale 42. With $2\log n$ bits/wedge, the total volume of wedge checks is 218PB and 52.8EB for RMAT graphs of scales 36 and 42, respectively\footnote{Throughout this paper, a petabyte (PB) is $2^{50}$ bytes and an exabyte (EB) is $2^{60}$ bytes.}.

 Beamer~\emph{et~al.} \cite{beamer2011searching} find a typical BFS on a scale 27 Graph500 RMAT graph has 7 levels, so 4 bits is a reasonable estimate for $\log D$ in our analyses of scale 36 and 42 graphs. 

\looseness=-1
The methodology for estimating the value of the covering ratio $c$ for RMAT graphs is as follows. RMAT graphs from scale 6 to 23 are generated, and the exact value of $c$ is determined for each by counting the horizontal-edges after a breadth-first search. The data fit to an exponential model $c = 1.1773 e ^{-0.036 \cdot \mbox{scale}}$ with very high $R^2 = 0.9956$ (see Fig.~\ref{fig:c}). For scale 36, $c$ is estimated to be 0.311 and for scale 42, $c$ is estimated to be 0.260.

In our new distributed-memory approach \emph{CETC-DM} for scale 36, where the communication cost is 
$m \cdot (\lceil \log D \rceil + (kp + 3) \lceil\log n \rceil) + (p-1) \lceil \log n \rceil$ bits. With $\lceil \log D \rceil = 4$, and assuming $p=128$ processors, we have a total communication volume of 192TB, for a communication reduction of $1156 \times$. 
For scale 42, and assuming $p=256$ processors, we estimate the communication of our new distributed-memory triangle counting algorithm \emph{CETC-DM} as 22.8PB, for a communication reduction of $2368 \times$.

\section{Conclusion}
In this paper we design and implement a novel, fast triangle counting algorithm \emph{CETC}, that uses new techniques called cover edge set, to improve the performance. It is the first algorithm in decades to shine a new light on triangle counting and use a wholly new method of cover-edge to reduce the work of set intersections, rather than other approaches that are variants of the well-known vertex-iterator and edge-iterator methods.  We provide extensive performance results for sequential triangle counting algorithms for sparse graphs in a uniform manner. Furthermore, we employ OpenMP to parallelize most of the sequential algorithms we implemented and investigate their performance.
The results use Intel's latest processor family, the Intel Sapphire Rapids (Platinum 8480+) launched in the 1st quarter of 2023.  The new triangle counting algorithm can benefit when the results of a BFS are available, which is often the case in network science.  Additionally, this work will inspire much interest within the Graph Challenge community to implement versions of the presented algorithms for large-shared memory, distributed memory, GPU, or multi-GPU frameworks.

\section{Reproducibility}
The triangle counting source code is open source and available on GitHub at \url{https://github.com/Bader-Research/triangle-counting}.  The input graphs are from the Stanford Network Analysis Project (SNAP) available from \url{http://snap.stanford.edu/}.

\begin{acks}
This research was funded in part by NSF grant number CCF-2109988.
\end{acks}

\bibliographystyle{ACM-Reference-Format}
\bibliography{ref}


\end{document}